\documentclass{article}

\RequirePackage{amsthm,amsmath,amsfonts,amssymb}
\RequirePackage[numbers]{natbib}

\usepackage{comment}
\usepackage[letterpaper, portrait, margin=0.75in]{geometry}

\usepackage{amsmath} %
\usepackage{amssymb}  %
\usepackage{bm}
\usepackage{bbm}
\usepackage{mathtools}
\usepackage{chemfig}
\mathtoolsset{showonlyrefs=true}
\usepackage{amsthm}

\theoremstyle{plain}
\newtheorem{lemma}{Lemma}[section]
\newtheorem{theorem}{Theorem}[section]
\newtheorem{corollary}{Corollary}[section]

\theoremstyle{remark}
\newtheorem{remark}{Remark}[section]

\newtheorem{assumption}{Assumption}[section]
\newtheorem{condition}{Condition}[section]
\newtheorem{definition}{Definition}[section]

\renewcommand\vec{\bm}

\newcommand\Lip{\mathrm{Lip}}

\newcommand{\Wass}{\mathcal{W}_1}

\newcommand{\cond}{\mathrm{cond}}

\DeclarePairedDelimiter\abs{\lvert}{\rvert}%
\DeclarePairedDelimiter\norm{\lVert}{\rVert}%

\DeclareMathOperator{\diag}{diag}

\newcommand{\al}{\alpha}
\newcommand{\bt}{\beta}
\newcommand{\ga}{\gamma}
\newcommand{\ep}{\epsilon}

\newcommand{\la}{\lambda}

\newcommand{\Om}{\Omega}

\newcommand{\grad}{\nabla}

\newcommand{\bbE}{\mathbb{E}}
\newcommand{\bbR}{\mathbb{R}}

\newcommand{\bbQ}{\mathbb{Q}}

\newcommand{\bbZ}{\mathbb{Z}}

\newcommand{\calA}{\mathcal{A}}

\newcommand{\calD}{\mathcal{D}}

\newcommand{\calH}{\mathcal{H}}

\newcommand{\calN}{\mathcal{N}}
\newcommand{\calO}{\mathcal{O}}

\newcommand{\calS}{\mathcal{S}}

\newcommand{\calX}{\mathcal{X}}

\newcommand{\Vvec}{\vec{V}}
\newcommand{\Uvec}{\vec{U}}
\newcommand{\Wvec}{\vec{W}}
\newcommand{\Xvec}{\vec{X}}
\newcommand{\Xvectil}{\tilde{\vec{X}}}

\newcommand{\Yvec}{\vec{Y}}

\newcommand{\Zvec}{\vec{Z}}

\newcommand{\cvec}{\vec{c}}

\newcommand{\lvec}{\vec{l}}

\newcommand{\svec}{\vec{s}}

\newcommand{\vvec}{\vec{v}}

\newcommand{\xvec}{\vec{x}}
\newcommand{\yvec}{\vec{y}}
\newcommand{\zvec}{\vec{z}}

\newcommand{\alvec}{\vec{\alpha}}

\newcommand{\tauvec}{\vec{\tau}}

\newcommand{\lavec}{\vec{\lambda}}

\newcommand{\gabar}{\bar{\ga}}

\newcommand{\labar}{\bar{\lambda}}

\newcommand{\lavecbar}{\bar{\vec{\lambda}}}

\newcommand{\xivec}{\vec{\xi}}

\author{Theodore W. Grunberg\thanks{Department of Electrical Engineering and Computer Science, MIT, Cambridge, MA 02139 USA 
  (\texttt{grunberg@mit.edu}).}
\and Domitilla Del Vecchio\thanks{Department of Mechanical Engineering, MIT, Cambridge, 
MA 02139 USA 
  (\texttt{ddv@mit.edu}).}
}

\begin{document}

\title{A Stein's Method Approach to the Linear Noise Approximation for Stationary Distributions of Chemical Reaction Networks}

\maketitle

\begin{abstract}
Stochastic Chemical Reaction Networks are continuous time Markov chain models that describe the time evolution of the molecular counts of species interacting stochastically via discrete reactions. Such models are ubiquitous in systems and synthetic biology, but often have a large or infinite number of states, and thus are not amenable to computation and analysis. Due to this, approximations that rely on the molecular counts and the volume being large are commonly used, with the most common being the Reaction Rate Equations and the Linear Noise Approximation. For finite time intervals, Kurtz established the validity of the Reaction Rate Equations and Linear Noise Approximation, by proving law of large numbers and central limit theorem results respectively. However, the analogous question for the stationary distribution of the Markov chain model has remained mostly unanswered, except for chemical reaction networks with special structures or bounded molecular counts. In this work, we use Stein's Method to obtain sufficient conditions for the stationary distribution of an appropriately scaled Stochastic Chemical Reaction Network to converge to the Linear Noise Approximation as the system size goes to infinity. Our results give non asymptotic bounds on the error in the Reaction Rate Equations and in the Linear Noise Approximation as applied to the stationary distribution. As a special case, we give conditions under which the global exponential stability of an equilibrium point of the Reaction Rate Equations is sufficient to obtain our error bounds, thus permitting one to obtain conclusions about the Markov chain by analyzing the deterministic Reaction Rate Equations.
\end{abstract}

\section{Introduction}\label{sec:introduction}
\subsection{Background and existing results}
Systems composed of a set of chemical species interacting via a finite set of reactions are most often modeled as a Stochastic Chemical Reaction Network (SCRN), a continuous time Markov chain (CTMC) model for how the molecular counts of each species change over time. 
Such models have found widespread use in systems and synthetic biology as a method to capture the counts of relevant biomolecules within a single cell, where due to the small numbers of molecules present, the effect of stochasticity in the system's evolution can be profound~\cite{wilkinson2009stochastic,tian2006stochastic,paulsson2000random,mcadams1997stochastic,arkin1998stochastic,elowitz2002stochastic,weinberger2005stochastic}.
The time evolution of the probability distribution of molecular counts can be computed via the Kolmogorov forward equations, often called the Chemical Master Equation \cite{gillespie1992rigorous}, or by a Monte Carlo approach using the Doob-Gillespie Algorithm \cite{doob1945markoff,gillespie1976general}. However, due to the very large, and often countably infinite, number of states in the Markov chain, simpler approximations have been proposed, the most prominent of which exploit the fact that for many physical systems the volume in which the reactions occur and the molecular counts of all species are large. These approximations are formalized by considering a family of Markov chains $\Xvec_\Om(t)$, where $\Xvec_\Om(t)$ is the vector of molecular counts at time $t$ and $\Om$ is the volume. The Reaction Rate Equation (RRE) model is an Ordinary Differential Equation (ODE), whose solution $\vvec(t)$ is a deterministic approximation for $\Xvec_\Om(t)/\Om$, rigorously justified by Kurtz on finite time intervals in the sense that under mild technical conditions, $\Xvec_\Om(t)/\Om \rightarrow_p \vvec(t)$ uniformly on any finite time interval \cite{kurtz1972relationship}. Due to $\vvec(t)$ being the solution to an ODE, the computational and analytical simplicity of the RRE model has lead to its widespread use in both modeling in systems biology \cite{Alon:2006aa} and design in synthetic biology \cite{BFS,gardner2000construction,elowitz2000synthetic}. In essence, the RRE model approximates the concentrations by a constant at each point in time, neglecting any stochastic fluctuation away from this value. In order to account for these fluctuations, Van Kampen expanded the Forward Kolmogorov Equations in powers of $1/\sqrt{\Om}$ \cite{kampen1961power,van1992stochastic}. The first order approximation he derived in this manner, called the Linear Noise Approximation (LNA), adds a correction term to the RREs, specifically a Gaussian approximation to fluctuations of $\Xvec_\Om(t)/\Om$ about $\vvec(t)$. The LNA was rigorously justified by Kurtz in the sense of giving technical conditions for the central limit theorem result $\frac{1}{\sqrt{\Om}}\left(\Xvec_\Om(t)-\Om \vvec(t)\right) \rightarrow_p \calN(0, P(t))$ to hold uniformly on finite time intervals with an appropriate covariance matrix $P(t)$. The LNA is similar to the RREs in that the approximating distribution can be easily computed from the solution to an ODE, and due to this simplicity it has been used widely to quantify noise in natural biological systems, \cite{paulsson2004summing,paulsson2005models,elf2003fast,tao2007effect}, as well as for design \cite{singh2009optimal,singh2011negative}, and inference \cite{komorowski2009bayesian,fearnhead2014inference,finkenstadt2013quantifying} in synthetic biology.

However, a key weakness of the theoretical justification for the RREs and the LNA is that the results given in \cite{kurtz1972relationship} do not say anything about the relationship between the stationary distribution of the Markov chain model and the steady state behavior of the RREs or LNA, even when the stationary distribution exists uniquely for all $\Om$ and the RREs/LNA has a well defined limiting behavior. In a later work, Kurtz gave a sufficient condition for the uniform in time convergence of $\Xvec_\Om(t)/\Om$ to $\vvec(t)$ and of $\frac{1}{\sqrt{\Om}}\left(\Xvec_\Om(t)-\Om \vvec(t)\right)$ to $\calN(0,P(t))$~\cite{kurtz1976limit}. However, that result requires that the state space of $\Xvec_\Om(t)/\Om$ be bounded uniformly in $\Om$, which rules out the analysis of many biological systems of interest where the molecular counts of certain species can grow unbounded, and thus for all $\Om$, $\Xvec_\Om(t)$ evolves on a countably infinite subset of the integer lattice. This technical difficulty in using the RREs or LNA to approximate the stationary distribution of an SCRN has long been understood \cite{van1992stochastic}, and in fact it is often the case that the Markov chain and RREs/LNA models have different long term behavior \cite{ioannis2008noise}.
In this work we give sufficient conditions under which there exists a constant $C$ such that for all sufficiently large $\Om$,
\begin{equation}\label{eq:general_result}
    \Wass\left(\frac{1}{\sqrt{\Om}}\left(\Xvec_\Om^\infty-\Om \vvec^*\right), \Yvec^\infty\right) \leq C\frac{\ln{\Om}}{\sqrt{\Om}},
\end{equation}
where $\Xvec_\Om^\infty$ is distributed according to the stationary distribution of $\Xvec_\Om(t)$, $\vvec^*$ is an appropriate equilbrium point of the RRE, $\Yvec^\infty \sim \calN(0, P^\infty)$ with $P^\infty$ a covariance matrix computed via the LNA, and $\Wass$ is the 1-Wasserstein distance between the laws of the two random variables. This result immediately implies both that the Markov chain's stationary distribution converges to an equilibrium point of the RRE in the sense $\Xvec_\Om^\infty/\Om \rightarrow_d \vvec^*$, where $\vvec^*$ is an appropriate equilbrium of the RRE, and that $\frac{1}{\sqrt{\Om}}\left(\Xvec_\Om(t)-\Om \vvec(t)\right) \rightarrow_d \calN(0, P^\infty)$~\cite{gibbs2002choosing}. We give sufficient conditions in terms of second moment of $\frac{1}{\sqrt{\Om}}\left(\Xvec_\Om^\infty-\Om \vvec^*\right)$ for \eqref{eq:general_result} to hold, and additionally give sufficient conditions on only the RRE model for \eqref{eq:general_result} to hold, which include global exponential stability of $\vvec^*$. Interestingly, the convergence rate given by \eqref{eq:general_result} is identical to the convergence rate shown by Kurtz on finite time intervals \cite{kurtz1978strong}.

The primary technique we use to prove our results is Stein's Method~\cite{chatterjee2014short}. Stein's Method is a powerful technique for proving central limit theorems introduced by Stein in 1972~\cite{stein1972bound,stein1986approximate}, and used in a wide variety of situations since then, such as when the limiting distribution is Poisson \cite{barbour1988stein} or Binomial \cite{ehm1991binomial}, as well as for convergence to a diffusion process \cite{barbour1990stein}. The most closely related work to our study comes from queueing theory, where Braverman et al. using Stein's Method, and Gurvich using a closely related technique, studied diffusion approximations of certain queueing systems in the heavy traffic regime~\cite{braverman2017stein,braverman2017steinMPH,gurvich2014diffusion}. Our application of Stein's method follows that of \cite{braverman2017stein} at a high level, where the key difference is that our family of CTMCs and limiting process is in $n$ dimensions instead of one. In order to handle this, we use results on Stein Factors from \cite{gorham2019measuring}.

\subsection{Related work}
Although rigorous work on approximating the stationary distribution of SCRNs in the large volume limit has been limited, a variety of related ideas have been presented in the literature. In particular, when the microscopic rates are all affine, it is known that the LNA matches the first and second moments of $\Xvec_\Om(t)$ for all time \cite{van1992stochastic,ioannis2008noise}. This idea has been extended to a restricted class of SCRNs with nonlinear propensities in \cite{grima2015linear-noise}. Additionally, Sontag and Singh gave a special class of SCRNs where the moments can be exactly computed even when they may not match the moments of the LNA exactly~\cite{sontag2015exact}. However, in these cases it is still unclear if the stationary distribution of $\frac{1}{\sqrt{\Om}}\left(\Xvec_\Om(t)-\Om v^*\right)$ converges to the Gaussian given by the LNA in the large volume limit. In fact, the results we present allow one to rigorously go from convergence of $\frac{1}{\sqrt{\Om}}\left(\Xvec_\Om^\infty-\Om v^*\right)$ to $\calN(0,P^\infty)$ with respect to the first and second moments, to convergence in 1-Wasserstein distance. We can thus strengthen the results of e.g. \cite{grima2015linear-noise} with limited additional work. We note that methods for establishing convergence of the stationary distributions of a family of Markov Processes to the stationary distribution of a limiting process are given in \cite{ethier2009markov}. However, not only have these methods not been successfully applied to SCRNs, but also our Stein's Method approach yields non-asymptotic bounds on the approximation error in 1-Wasserstein distance.

In certain special cases, such as detailed or complex balance, the stationary distribution of $\Xvec_\Om(t)$ can be computed in closed form, and thus in principle one could directly analyze the approximation error when using the LNA~\cite{whittle1986systems,anderson2010product,meng2017recursively}. However, for most SCRNs of interest in biology, such results are not applicable. For example, a simple two step transcription-translation model for protein expression is not detailed or complex balanced \cite{BFS}, and thus the results in \cite{whittle1986systems,anderson2010product} cannot be used to obtain a closed form expression for the stationary distribution, whereas the results of \cite{meng2017recursively} are only exact for systems where $\Xvec_\Om(t)$ has finitely many states.

We note that the LNA is only one diffusion approximation for SCRNs, and in fact the Chemical Langevin Equation (CLE) is often a more accurate approximation to $\Xvec_\Om(t)/\Om$ for moderately large volumes~\cite{gillespie2000chemical}. The convergence of $\Xvec_\Om(t)/\Om$ to the CLE on finite time intervals was established by Kurtz \cite{kurtz1978strong}, but using the CLE to approximate the stationary distribution of $\Xvec_\Om(t)/\Om$ has remained mostly unstudied. In fact, the CLE often exhibits finite time breakdown and thus does not have a well defined stationary distribution, a problem that has been tackled by allowing complex values in the Stochastic Differential equation~\cite{schnoerr2014complex}, as well as imposing the physical constraint that concentrations are nonnegative \cite{leite2019constrained}.

\subsection{Outline of the rest of the paper}
The rest of the paper is as follows: In Section \ref{sec:problem_setting} we give notation and mathematical background, and then introduce the three relevant models. In Section \ref{sec:main_results} we state our main results, the proofs of which appear in Section \ref{sec:main_proof}. In Section \ref{sec:examples} we illustrate our results with several examples. In Section \ref{sec:conclusion} we provide concluding remarks and future research directions.
\section{Problem Setting}\label{sec:problem_setting}
In this section we first introduce the notation we will use, and then introduce SCRNs and the associated CTMC model along with and the Reaction Rate Equations and Linear Noise Approximation models. We also introduce the 1-Wasserstein distance, which we will use to measure the distance between probability distributions.
\subsection{Notation}
We denote (column) vectors by bold symbols $\xvec,\Zvec$, with the elements indexed as $\xvec = [x_1,\dots,x_n]^T$. Let $\bbZ_{\geq 0}^n$ [$\bbZ_{>0}^n$] denote the nonnegative [positive] integer lattice in $n$ dimensions. Let $\bbQ$ denote the set of rational numbers. Let $\bbR^n$, $\bbR_{\geq 0}^n$, and $\bbR^n_{>0}$ denote n-dimensional Euclidian space, and the nonnegative and positive orthants of the same respectively. For $\xvec \in \bbR^n$ we denote the standard 2-norm by $\norm{\xvec}$, and for a matrix $A\in \bbR^{n\times m}$ we denote by $\norm*{A}$ the induced 2-norm $\norm*{A} = \sup_{\xvec \neq 0} \frac{\norm*{A\xvec}}{\norm{\xvec}}$, and denote by $\mathrm{cond}(A)$ the condition number of $A$. When all eigenvalues of $A\in \bbR^{n\times n}$ have strictly negative real parts we say that $A$ is Hurwitz. Given a function $g:\bbR^n \rightarrow\bbR$, we denote by $\frac{\partial g}{\partial \zvec}$, $\grad g$, and $\grad^2g$ the row vector of partial derivatives, the gradient, and the Hessian of $g$ respectively. We denote the third (directional) derivative of $g$ by $\grad^3g[\xvec_1,\xvec_2,\xvec_3]$, along the directions $\xvec_1,\xvec_2,\xvec_3$, and define its induced norm as $\sup_{\norm{\xvec_1}=\norm{\xvec_2}=\norm{\xvec_3}=1}\norm{\grad^3g [\xvec_1,\xvec_2,\xvec_3]}$. For two functions $f$ and $g$ from $\bbR^n$ to $\bbR$, we define the convolution of $f$ and $g$ as $(f\star g)(\xvec) = \int_{\bbR^n} f(\xvec-\svec)g(\svec)\mu(d\svec)$, with $\mu$ the Lebesgue measure on $\bbR^n$. We denote the set of k-continuously differentiable functions on a domain $\calD$ by $C^k(\calD)$, where we omit the domain if it is clear from context.

\subsection{Stochastic Chemical Reaction Networks}\label{sec:SCRN}
We consider Stochastic Chemical Reaction Networks parameterized by $\Om\in\calO \subseteq \bbR_{>0}$, composed of $n$ species $\mathrm{X}_1,\mathrm{X}_2,\dots,\mathrm{X}_n$, interacting according to $r$ reactions, each characterized by a nonzero reaction vector $\xivec_i \in \bbZ^n$ and a microscopic propensity function $\la_i(\xvec)$, which depends on $\Om$. The reaction vector describes the change in the counts on the species when reaction $i$ fires, and the microscopic propensity function describes the rate at which reaction $i$ fires. Thus, for each $\Om \in \calO$, the Stochastic Chenical Reaction Network describes a CTMC $\{\Xvec_\Om(t),t \geq 0\}$, where $\Xvec_\Om(t) = [X_1(t),X_2(t),\dots,X_n(t)]^T \in \calX_\Om \subseteq \bbZ_{\geq 0}^n$ represents the molecular counts of the species. The state space $\calX_\Om$ can be chosen as any subset of $\bbZ_{\geq 0}^n$ such that 1) $\la_i(\xvec) \geq 0$ for all $\xvec \in \calX_\Om$ and 2) for all $\xvec \in \calX_\Om$ and all $i\in \{1,\dots,r\}$ such that $\xvec+\xivec_i \notin \calX_\Om$, $\la_i(\xvec)=0$. The infinitesimal transition rates from $\xvec$ to $\xvec'$ are given by $q_{\xvec,\xvec'} = \sum_{i:\xvec+\xivec_i=\xvec'}\la_i(\xvec)$. We refer the interested reader to \cite{anderson2015stochastic} for a more thorough background on SCRNs. Throughout this work we will need to impose the following assumption on our SCRNs:
\begin{assumption}\label{assum:propensities}
The propensity functions have the form
\begin{equation}
    \la_i(\xvec) = \Om \labar_i\left(\frac{\xvec}{\Om}\right).
\end{equation}
Additionally, there exists $\calX = \prod_{j=1}^n I_j \subseteq \bbR_{\geq 0}^n$ where for each $j$, $I_J = [0,x_j^{max}]$ for some $x_j^{max} \in \bbQ_{>0}$, or $I_j = [0,\infty)$, such that:
\begin{enumerate}
    \item For all $i\in\{1,\dots,r\}$, $\labar_i(\vvec) >0$ for all $\vvec \in \calX$ such that there exists $\eta >0$ satisfying $\vvec + \eta \xivec_i \in \calX$.
    \item For all $i\in\{1,\dots,r\}$, $\labar_i(\vvec) = 0$ for all $\vvec \in \calX$ such that for all $\eta >0$, $\vvec + \eta \xivec_i \notin \calX$.
\end{enumerate}
\end{assumption}
\begin{remark}
    The functions $\labar_i(\vvec)$ are the \emph{macroscopic propensities} of each reaction.
\end{remark}
Under Assumption \ref{assum:propensities}, we can define $\calO = \left\{\Om \in \bbR_{>0} \middle| \forall j,\; \Om x_j^{max} \in \bbZ_{> 0} \right\}$. Note that our assumption that $x_j^{max} \in \bbQ$ ensures that $\sup \calO = +\infty$, i.e. our family of Markov chains includes those with arbitrarily large volume. Under Assumption \ref{assum:propensities} we will always choose to define the state space of our CTMC as $\calX_\Om = \bbZ_{\geq 0}^n \cap \Om \calX$. In order to apply our Stein's Method technique, we will need the following assumption:
\begin{assumption}\label{assum:C2b}
    For all $i \in \{1,\dots,r\}$, $\labar_i(\vvec)$ is twice continuously differentiable on some open set $\calD$ containing $\calX$, and has bounded second derivatives on $\calX$.
\end{assumption}
One form of $\la_i(\xvec)$ which satisfies Assumptions \ref{assum:propensities} and \ref{assum:C2b} is mass action kinetics \cite{van1992stochastic}, which means that the propensities are of the form
\begin{equation}\label{eq:masssaction}
    \la_i(\xvec) = k_i\Om^{1-\sum_j \xi_{ri,j}} \prod_{j=1}^n x_j^{\xi_{ri,j}},
\end{equation}
where $\xivec_i = \xivec_{pi} - \xivec_{ri}$ with $\xivec_{pi},\xivec_{ri} \in \bbZ_{\geq 0}^n$ is the reaction vector of the chemical reaction
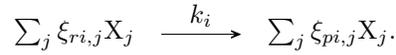
\begin{equation}
    \schemestart
    \subscheme{$\sum_{j}\xi_{ri,j}\mathrm{X}_j$}
    \arrow(x1--x1){->[$k_i$]}[0]
    \subscheme{$\sum_{j}\xi_{pi,j}\mathrm{X}_j$}
    \schemestop.
\end{equation}
The requirement that $\sum \xivec_{pi} \leq 2$ must be imposed for $\la_i(\xvec)$ to satisfy Assumption \ref{assum:C2b}. See \cite{gillespie1992rigorous} for a discussion of the physical assumptions necessary to reach such propensities. We note that \eqref{eq:masssaction} requires that no reactant appears twice on the left hand side of the reaction. Having a reactant $\mathrm{X}_j$ appear twice on the left hand side of a reaction creates a propensity proportional to $x_j^2$ instead of the physically accurate $x_j(x_j-1)$ \cite{gillespie1992rigorous}. Here $k_i$ is the reaction rate constant of reaction $i$ and $\xi_{ri,j}$ is the $j$\textsuperscript{th} element of $\xivec_{ri}$. Another form of $\la_i(\xvec)$ that satisfies Assumptions \ref{assum:propensities} and \ref{assum:C2b} is
\begin{equation}
    \la_i(\xvec) = \Om\frac{k^i_1 (x_{j_i}/\Om)^p}{k^i_2 + (x_{j_i}/\Om)^q},
\end{equation}
where $j_i \in \{1,\dots,n\}$ and $p,q \in \bbZ_{>0}$ such that $p+q \leq 3$ or $p=q$. Such ``Hill-type propensities'' can arise through timescale separation in systems of chemical reactions with mass action kinetics \cite{kang2013separation,melykuti2014equilibrium,hirsch2023error}. For simplicity, we will impose the following assumption:
\begin{assumption} \label{assum:stoich_full_dim}
    The stoichiometric subspace, $\mathrm{span} \{\xivec_1,\dots,\xivec_r\}$, has dimension $n$, and for each $\Om \in \calO$, $\Xvec_\Om(t)$ is irreducible.
\end{assumption}
Though Assumption \ref{assum:stoich_full_dim} appears to rule out SCRNs that have conservation laws, i.e. a linear combination of species counts that remains unchanged by all reactions, we will show through an example in Section \ref{sec:examples} that Assumption \ref{assum:stoich_full_dim} does not fundamentally prevent the analysis of SCRNs with conservation laws, as long as a change of coordinates exist where the system can be represented in a lower dimensional lattice $\bbZ_{\geq 0}^{n'}$ as an SCRN with stoichiometric subspace of dimension $n'$. For $f:\bbR^n\rightarrow\bbR$ the generator of $\Xvec_\Om(t)$ is
\begin{equation}
    G_{\Xvec_\Om}f(\xvec) = \sum_{i=1}^r \la_i(\xvec) \left(f(\xvec+\xivec_i)-f(\xvec)\right).
\end{equation}
We will also consider two shifted and scaled version of $\Xvec_\Om(t)$ defined by
\begin{equation}
    \Xvectil_\Om (t) = \frac{1}{\sqrt{\Om}}\left(\Xvec_\Om(t) - \Om \vvec^*\right),
\end{equation}
and
\begin{equation}
    \tilde{\Xvectil}_\Om(t) = \frac{1}{\Om}\Xvec_\Om(t) - \vvec^*,
\end{equation}
where $\vvec^*$ is a specified point in $\calX$. For $f:\bbR^n\rightarrow\bbR$ the generator of $\Xvectil_\Om(t)$ is given by
\begin{equation}\label{eq:scaled_generator}
    G_{\Xvectil_\Om}f(\xvec) = \sum_{i=1}^r \la_i(\sqrt{\Om}\xvec + \Om \vvec^*) \left(f(\xvec+\frac{1}{\sqrt{\Om}}\xivec_i)-f(\xvec)\right),
\end{equation}
and the generator of $\tilde{\Xvectil}_\Om(t)$ is given by
\begin{equation}
    G_{ \tilde{\Xvectil}_\Om}f(\xvec) = \sum_{i=1}^r \Om\labar_i(\xvec + \vvec^*) \left(f(\xvec+\frac{1}{\Om}\xivec_i)-f(\xvec)\right).
\end{equation}
When $\Xvec_\Om(t)$ has a unique stationary distribution $\pi$, we denote by $\Xvec_\Om^\infty$ a random variable having law $\pi$. We likewise define $\Xvectil_\Om^\infty$ and $\tilde{\Xvectil}_\Om^\infty$.

\subsection{Reaction Rate Equations}
The Reaction Rate Equations (RREs) are an Ordinary Differential Equation Model which is used to approximate $\Xvec_\Om(t)/\Om$. Defining $\lavecbar(\vvec) = \begin{bmatrix} \labar_1(\vvec), & \labar_2(\vvec), & \dots,& \labar_r(\vvec)\end{bmatrix}^T$ and $S = 
    \begin{bmatrix}
    \xivec_{1} & \xivec_2 & \dots & \xivec_r
    \end{bmatrix}$, the reaction rate equations are
\begin{align}\label{eq:RRE}
    \frac{d}{dt}\vvec(t) & = F(\vvec(t)), & \vvec(0) & = \vvec_0
\end{align}
where $F(\vvec) = S \lavecbar(\vvec)$ and $\vvec \in \calX \subseteq \bbR^n_{\geq 0}$. An equilibrium point of the RREs is any point $\vvec^* \in \calX$ satisfying $0 = F(\vvec^*)$. We note that under Assumption \ref{assum:propensities}, $\calX$ will be positive invariant with respect to \ref{eq:RRE}.

\subsection{Linear Noise Approximation}
The Linear Noise Approximation (LNA) is a diffusion approximation obtained by expanding the Chemical Master Equation using the ansatz $\Xvec_\Om(t) \approx \Om \vvec(t) + \sqrt{\Om}\Yvec'(t)$ \cite{van1992stochastic}. The terms $\vvec(t)$ and $\Yvec'(t)$ are given by
\begin{subequations}\label{eq:LNA}
\begin{align}
    \frac{d}{dt}\vvec(t) & = F(\vvec(t)), & \vvec(0) & = \vvec_0,\label{eq:LNA_RRE} \\
    d\Yvec'(t) & = \frac{\partial F(\vvec)}{\partial \vvec} \Yvec'(t) dt + S \diag \sqrt{\lavecbar(\vvec)}d\Wvec(t), & \Yvec'(0) & = \Yvec'_0,\label{eq:LNA_FLUC}
\end{align}
\end{subequations}
\noeqref{eq:LNA_RRE}\noeqref{eq:LNA_FLUC}where evidently \eqref{eq:LNA_RRE} is the RRE \eqref{eq:RRE}. Here $\Wvec(t)$ is a unit covariance Wiener Process. Let $\vvec^*$ be an equilibrium point of \eqref{eq:LNA_RRE}. We note that $\labar_i(\vvec) \geq 0$ must hold for \eqref{eq:LNA_FLUC} to make sense, which is guaranteed by e.g. Assumption \ref{assum:propensities}. We denote by $\Yvec(t)$ the solution to \eqref{eq:LNA_FLUC} with $\vvec(t) = \vvec^*$ and $\Yvec'(0) = \Yvec_0$. The generator of $\Yvec(t)$ is
\begin{equation}
    G_{\Yvec}f(\xvec) = \frac{\partial f(\xvec)}{\partial \xvec}\cdot \frac{\partial F(\vvec^*)}{\partial \vvec}\xvec + \frac{1}{2}\sum_{i=1,j=1}^n D_{ij} \frac{\partial^2 f(\xvec)}{\partial x_i \partial x_j},
\end{equation}
where $D =  S \diag \left(\lavecbar(\vvec^*)\right) S^T$. We remark that $\Yvec(t)$ is an Ornstein–Uhlenbeck process, and thus $\Yvec(t)$ is Gaussian as long as $\Yvec_0$ is. The stationary distribution of $\Yvec(t)$ is a zero mean Gaussian with covariance matrix $P^\infty \in \bbR^{n\times n}$ given as the solution to the Lyapunov Equation
\begin{equation}
    \frac{\partial F(\vvec^*)}{\partial \vvec}P^\infty + P^\infty \frac{\partial F(\vvec^*)}{\partial \vvec}^T = - S \diag \lavecbar(\vvec^*) S^T,
\end{equation}
which has a unique positive definite solution as long as $\frac{\partial F(\vvec^*)}{\partial \vvec}$ is Hurwitz. In this case, we denote by $\Yvec^\infty$ a random variable distributed according to $\calN(0,P^\infty)$.

\subsection{1-Wasserstein Distance}
In this work we measure the distance between probability distributions using the 1-Wasserstein distance \cite{gibbs2002choosing}. By a slight abuse of notation, we define the 1-Wasserstein distance between two random variables as the 1-Wasserstein distance between the laws of the random variables, formally:
\begin{definition}
    Given two random variables $\Zvec_1$ and $\Zvec_2$, both taking values in $\bbR^n$, we define the \emph{1-Wasserstein distance} between $\Zvec_1$ and $\Zvec_2$ as
    \begin{equation}
        \Wass\left(\Zvec_1,\Zvec_2\right) = \sup_{h\in\Lip(1)} \abs*{\bbE\left[h\left(\Zvec_1\right)\right] - \bbE\left[h\left(\Zvec_2\right)\right]},
    \end{equation}
    where $\Lip(1) = \left\{h:\bbR^n\rightarrow\bbR\middle|\forall \xvec,\xvec'\in\bbR^n,\;\norm*{h(\xvec)-h(\xvec')}\leq \norm*{\xvec-\xvec'}\right\}$.
\end{definition}
We remark that this definition of 1-Wasserstein distance uses the dual formulation \cite{villani2009optimal}, and is a slight abuse of notation since the 1-Wasserstein distance is usually defined between two probability distributions. Our definition is equivalent to taking the 1-Wasserstein distance between the laws of $\Zvec_1$ and $\Zvec_2$.
\section{Approximating Stationary Distributions of SCRNs}\label{sec:main_results}
Here we state the main results of this work. To begin, we present the technical conditions that will guarantee that the stationary distribution of $\Xvectil_\Om(t)$ converges to $\calN(0,P^\infty)$. We then give our main result, showing that when the first and second moments of $\Xvectil_\Om^\infty$ are controlled, we can bound the 1-Wasserstein distance between $\Xvectil_\Om^\infty$ and $\Yvec^\infty$. We next show that this result can be used to analyze the RRE approximation as well. Finally, we give a method to check that the second moment of $\Xvectil_\Om^\infty$ is controlled by relating the global stability properties of the RRE to the required moment bound.

\begin{condition}\label{cond:k_1_moment_finite}
    There exists $p(\xvec)$, a polynomial of degree $\kappa$, such that for all $i$ and all $\xvec \in \calX$, $\labar_i(\xvec) \leq p(\xvec)$. Furthermore, there exists $\Om_0$ such that for all $\Om\in \left\{\Om \in \calO \middle| \Om \geq \Om_0\right\}$, $\Xvectil_\Om(t)$ has a unique stationary distribution $\tilde{\pi}$, which satisfies
    \begin{equation}
        \bbE_{\Xvectil_\Om^\infty \sim \tilde{\pi}}\left[\norm{\Xvectil_\Om^\infty}^{\kappa+1}\right] < \infty.
    \end{equation}
\end{condition}
Condition \ref{cond:k_1_moment_finite} does not require any type of uniformity in $\Om$, and therefore the condition $\bbE\left[\norm{\Xvec_\Om^\infty}^{\kappa+1}\right] < \infty$ can be checked instead.
\begin{condition}[Uniform moment bounds]\label{cond:moment}
    There exist constants $M_1$, $M_2$, and $\Om_0$ such that for all $\Om\in \left\{\Om \in \calO \middle| \Om \geq \Om_0\right\}$, $\Xvectil_\Om(t)$ has a unique stationary distribution $\tilde{\pi}$, which satisfies
    \begin{equation}\label{eq:moment_bound_first}
        \bbE_{\Xvectil_\Om^\infty \sim \tilde{\pi}}\left[ \norm{\Xvectil_\Om^\infty}\right] \leq M_1,
    \end{equation}
    and
    \begin{equation}\label{eq:moment_bound_second}
        \bbE_{\Xvectil_\Om^\infty \sim \tilde{\pi}}\left[ \norm{\Xvectil_\Om^\infty}^2\right] \leq M_2.
    \end{equation}
\end{condition}

\begin{remark}\label{rem:just_second_moment}
In Condition \ref{cond:moment}, the bound on $\bbE_{\Xvectil_\Om^\infty \sim \tilde{\pi}}\left[ \norm{\Xvectil_\Om^\infty}^2\right]$ implies the existence of a bound on $\bbE_{\Xvectil_\Om^\infty \sim \tilde{\pi}}\left[ \norm{\Xvectil_\Om^\infty}\right]$. This can be seen from the fact that by Jensen's inequality, $\bbE_{\Xvectil_\Om^\infty \sim \tilde{\pi}}\left[ \norm{\Xvectil_\Om^\infty}\right]^2 \leq \bbE_{\Xvectil_\Om^\infty \sim \tilde{\pi}}\left[ \norm{\Xvectil_\Om^\infty}^2\right]$, and thus $\bbE_{\Xvectil_\Om^\infty \sim \tilde{\pi}}\left[ \norm{\Xvectil_\Om^\infty}^2\right] \leq M_2$ implies that $\bbE_{\Xvectil_\Om^\infty \sim \tilde{\pi}}\left[ \norm{\Xvectil_\Om^\infty}\right] \leq \sqrt{M_2}$. However, we state Condition \ref{cond:moment} as it is to allow for the possibility of using a stronger bound on $\bbE_{\Xvectil_\Om^\infty \sim \tilde{\pi}}\left[ \norm{\Xvectil_\Om^\infty}\right]$ than the one derived from $M_2$.
\end{remark}

\begin{remark}\label{rem:concentration_moments}
    Recalling that $\tilde{\Xvectil}_\Om = \frac{1}{\Om}\left(X_\Om - \Om \vvec^*\right)$ is the ``concentration scaled'' Markov chain shifted to $\vvec^*$, Condition \ref{cond:moment} is equivalent to the existence of $\Om_0$, $M_1$, and $M_2$ such that for all $\Om \in \left\{\Om\in\calO\middle|\Om \geq \Om_0\right\}$, there exists a unique stationary distribution of $\tilde{\Xvectil}_\Om(t)$, $\tilde{\tilde{\pi}}$, and it satisfies
    \begin{equation}
        \bbE_{\tilde{\Xvectil}_\Om^\infty \sim \tilde{\tilde{\pi}}}\left[ \norm{\tilde{\Xvectil}_\Om^\infty}\right] \leq \frac{1}{\sqrt{\Om}}M_1,
    \end{equation}
    and
    \begin{equation}
        \bbE_{\tilde{\Xvectil}_\Om^\infty \sim \tilde{\tilde{\pi}}}\left[ \norm{\tilde{\Xvectil}_\Om^\infty}^2\right] \leq \frac{1}{\Om}M_2.
    \end{equation}
\end{remark}
We are now ready to state our main result.
\begin{theorem}\label{thm:moments_imply_convergence}
    Consider an SCRN $\Xvec_\Om (t)$ satisfying Assumptions \ref{assum:propensities}, \ref{assum:C2b}, and \ref{assum:stoich_full_dim} and let $\vvec^* \in \mathrm{int}(\calX)$ be an equilibrium point of \eqref{eq:LNA_RRE} such that $\frac{\partial F}{\partial \vvec}(\vvec^*)$ is Hurwitz. If Conditions \ref{cond:k_1_moment_finite} and \ref{cond:moment} hold, then there exist $C$ and $\Omega_0'$ such that for all $\Om\in \left\{\Om \in \calO \middle| \Om \geq \Om_0'\right\}$,
    \begin{equation}
        \Wass\left(\tilde{\Xvec}_\Om^\infty, \Yvec^\infty\right)
        \leq C \frac{\ln\Omega}{\sqrt{\Omega}}.
    \end{equation}
\end{theorem}

\begin{proof}
    See Section \ref{sec:proof_main_thm}.
\end{proof}
\begin{remark}
    The conclusion of Theorem \ref{thm:moments_imply_convergence} implies that
    \begin{equation}
        \lim_{\Omega\rightarrow\infty}\Wass\left(\tilde{\Xvec}_\Om^\infty, \Yvec^\infty\right) = 0,
    \end{equation}
    from which we can conclude convergence in distribution of $\tilde{\Xvec}_\Om^\infty$ to $\Yvec^\infty$~\cite{villani2009optimal}.
\end{remark}
\begin{remark}
    The constants $C$ and $\Omega_0'$ in Theorem \ref{thm:moments_imply_convergence} depend on the propensities and stoichiometry vectors, as well as $M_1$, $M_2$, and $\Omega_0$ from Condition \ref{cond:moment}. We note that an explicit computation of $C$ and $\Omega_0'$ is possible, as will be seen in the proof, where we give an explicit upper bound that is dominated by a term proportional to $\frac{\ln \Omega}{\sqrt{\Omega}}$.
\end{remark}
Though our Stein's Method approach does not require the preliminary step of establishing a law of large numbers type result relating $\tilde{\Xvectil}_\Om^\infty$ to $\vvec^*$, an equilibrium point of \eqref{eq:LNA_RRE}, we can conclude such a result whenever the assumptions of Theorem \ref{thm:moments_imply_convergence} are satisfied.
\begin{corollary}\label{coro:conc_wass_dist}
    Consider an SCRN $\Xvec_\Om (t)$ satisfying Assumptions \ref{assum:propensities}, \ref{assum:C2b}, and \ref{assum:stoich_full_dim} and let $\vvec^* \in \mathrm{int}(\calX)$ be an equilibrium point of \eqref{eq:LNA_RRE} such that $\frac{\partial F}{\partial \vvec}(\vvec^*)$ is Hurwitz. If Conditions \ref{cond:k_1_moment_finite} and \ref{cond:moment} hold, then there exist $\tilde{C}$ and $\tilde{\Omega}_0'$ such that for all $\Om\in \left\{\Om \in \calO \middle| \Om \geq \tilde{\Om}_0'\right\}$,
    \begin{equation}
        \Wass\left(\tilde{\Xvectil}_\Om^\infty,0\right)
        \leq \tilde{C} \frac{1}{\sqrt{\Omega}}.
    \end{equation}
\end{corollary}

\begin{proof}
    From the triangle inequality, we have that
    \begin{equation}\label{eq:wass_tri_X_Y_0}
        \Wass\left(\tilde{\Xvectil}_\Om^\infty,0\right) \leq \Wass\left(\tilde{\Xvectil}_\Om^\infty,\frac{1}{\sqrt{\Om}}\Yvec^\infty\right) + \Wass\left(\frac{1}{\sqrt{\Om}}\Yvec^\infty, 0\right).
    \end{equation}
    To bound the right hand side, notice that $\tilde{\Xvectil}^\infty_\Om = \frac{1}{\sqrt{\Om}}\Xvectil^\infty_\Om$. Thus,
    \begin{equation}\label{eq:W1_scaling}
        \Wass\left(\tilde{\Xvectil}_\Om^\infty,\frac{1}{\sqrt{\Om}}\Yvec^\infty\right) = \frac{1}{\sqrt{\Om}}\Wass\left(\Xvectil^\infty_\Om, \Yvec^\infty\right).
    \end{equation}
    Under the assumptions of the corollary, we can apply Theorem \ref{thm:moments_imply_convergence} to conclude that there exist $C$ and $\Om_0'$ such that for all $\Om\in \left\{\Om \in \calO \middle| \Om \geq \Om_0'\right\}$,
    \begin{equation}
        \Wass\left(\tilde{\Xvec}_\Om^\infty, \Yvec^\infty\right)
        \leq C \frac{\ln(\Omega)}{\sqrt{\Omega}},
    \end{equation}
    and thus from \eqref{eq:W1_scaling}, we have that for all $\Om\in \left\{\Om \in \calO \middle| \Om \geq \Om_0'\right\}$,
    \begin{equation}\label{eq:wass_conc_bound}
        \Wass\left(\tilde{\Xvectil}_\Om^\infty,\frac{1}{\sqrt{\Om}}\Yvec^\infty\right) \leq C\frac{\ln \Om}{\Om}.
    \end{equation}
    Observing that $\Wass\left(\frac{1}{\sqrt{\Om}}\Yvec^\infty, 0\right) = \frac{1}{\sqrt{\Om}}\Wass\left(\Yvec^\infty, 0\right)$, we have from \eqref{eq:wass_tri_X_Y_0} and \eqref{eq:wass_conc_bound} that for all $\Om\in \left\{\Om \in \calO \middle| \Om \geq \Om_0'\right\}$,
    \begin{align}
        \Wass\left(\tilde{\Xvectil}_\Om^\infty,0\right) & \leq C\frac{\ln \Om}{\Om} + \frac{1}{\sqrt{\Om}}\Wass\left(\Yvec^\infty, 0\right),\\
        & \leq \left(C+\Wass\left(\Yvec^\infty, 0\right)\right)\frac{1}{\sqrt{\Om}},
    \end{align}
    which completes the proof.
\end{proof}
A special case where Conditions \ref{cond:k_1_moment_finite} and \ref{cond:moment} are easy to verify is an SCRN with mass action kinetics and only zeroth and first order reactions. This idea is formalized in Corollary \ref{coro:zeroth_first_only}.
\begin{corollary}\label{coro:zeroth_first_only}
    Consider an SCRN $\Xvec_\Om(t)$ satisfying Assumption \ref{assum:stoich_full_dim} where for all $i$, $\labar_i(\xvec)$ is either of the form $\labar_i(\xvec) = k_i$ or $\labar_i(\xvec) = k_i x_{j_i}$ for some $k_i >0$ and $j_i \in \{1,\dots,n\}$. Let $\vvec^* \in \mathrm{int}(\calX)$ be an equilibrium point of \eqref{eq:LNA_RRE} such that $\frac{\partial F(\vvec^*)}{\partial \vvec}$ is Hurwitz. Then, there exist $C$ and $\Om_0'$ such that for all $\Om \geq \Om_0'$,
    \begin{equation}
        \Wass\left(\tilde{\Xvec}_\Om^\infty, \Yvec^\infty\right)
        \leq C \frac{\ln\Omega}{\sqrt{\Omega}}.
    \end{equation}
\end{corollary}
\begin{proof}
    The proof relies on the fact that for SCRNs with only zeroth and first order reactions, the LNA gives exactly the first and second moments of $\Xvec_\Om^\infty$. Specifically, the fact that $\frac{\partial F(\vvec^*)}{\partial \vvec}$ is Hurwitz implies that we can apply Proposition 7 in \cite{gupta2014scalable} to conclude that all moments of $\Xvec_\Om^\infty$ are finite, and thus Condition \ref{cond:k_1_moment_finite} holds. Additionally, since all of the propensities are affine, a computation of the first and second moment dynamics of $\Xvec_\Om(t)$ reveals that for all time~\cite{ioannis2008noise}, $\bbE\left[\tilde{\Xvectil}_\Om(t)\right] = 0$ and $\bbE\left[ \tilde{\Xvectil}_\Om(t) \tilde{\Xvectil}_\Om^T(t) \right] = \frac{1}{\Om}\bbE\left[\Yvec(t)\Yvec^T(t)\right]$. Therefore, by Remarks \ref{rem:just_second_moment} and \ref{rem:concentration_moments}, Condition \ref{cond:moment} is satisfied. The result then follows from Theorem \ref{thm:moments_imply_convergence}.
\end{proof}
We now give a result that gives sufficient conditions on the RREs for the conclusion of Theorem \ref{thm:moments_imply_convergence} to hold. This result is substantially easier to use in practice than Theorem \ref{thm:moments_imply_convergence}, since it requires only analyzing an ODE with $n$ variables, and not the stationary distribution of a CTMC over $\calX_\Om$.
\begin{theorem}\label{thm:RRE_GES}
    Consider an SCRN $\Xvec_\Om(t)$ satisfying Assumptions \ref{assum:propensities}, \ref{assum:C2b}, and \ref{assum:stoich_full_dim}, and let $\vvec^* \in \mathrm{int}(\calX)$. Suppose that there exist $K, \ga_1 >0$ such that for all $\vvec_0 \in \calX$,
    \begin{equation}
        \norm*{\vvec(t)-\vvec^*} \leq K e^{-\ga_1 t}\norm*{\vvec_0-\vvec^*},
    \end{equation}
    where $\vvec(t)$ is the solution to \eqref{eq:LNA_RRE}, and that there exists $\cvec \in \bbR^n_{>0}$ and $d,\ga_2>0$ such that for all $\vvec \in \left\{\vvec \in \calX \middle| \cvec^T\vvec \geq d \right\}$,
    \begin{equation}
        \cvec^TF(\vvec) \leq -\ga_2 \cvec^T \vvec.
    \end{equation}
    Then, there exist $C$ and $\Omega_0'$ such that for all $\Om\in \left\{\Om \in \calO \middle| \Om \geq \Om_0'\right\}$,
    \begin{equation}
        \Wass\left(\tilde{\Xvec}_\Om^\infty, \Yvec^\infty\right)
        \leq C \frac{\ln\Omega}{\sqrt{\Omega}}.
    \end{equation}
\end{theorem}
\begin{proof}
    See Section \ref{sec:RRE_GES_proof}.
\end{proof}
\section{Proofs of the Main Results}\label{sec:main_proof}
In this section we provide the proofs of the results presented in Section \ref{sec:main_results}.
\subsection{Proof of Theorem \ref{thm:moments_imply_convergence}}\label{sec:proof_main_thm}
We require the following Lemma from \cite{glynn2008bounding}, which gives sufficient conditions for the the Basic Adjoint Relationship (BAR) to hold.
    \begin{lemma}[Proposition 3 in \cite{glynn2008bounding}]\label{lem:glynn}
        Consider a CTMC over state space $\calS$ with rate matrix having entries $G(\xvec,\xvec')$ and stationary distribution $\pi(\xvec)$. If
        \begin{equation}
            \sum_{\xvec\in \calS} \pi(\xvec)\abs*{G(\xvec,\xvec)}\abs*{f(\xvec)} < \infty,
        \end{equation}
        then
        \begin{equation}
            \bbE_{\Xvec^\infty \sim \pi}\left[Gf(\Xvec^\infty)\right] = 0.
        \end{equation}
    \end{lemma}
    It will be convenient for us to consider the 1-Wasserstein distance computed by taking a supremum over not the 1-Lipschitz functions, but the 1-Lipschitz functions that are additionally in $C^3$ with bounded derivatives. To this end, let
    \begin{equation}
        \calH = \{h \in C^3(\bbR^n)| h\in \Lip(1),\;\forall k \in \{1,2,3\},\;\sup_{\xvec\in\bbR^n}\norm{\grad^k h(\xvec)} < \infty \}.
    \end{equation}
    We have the following Lemma, which is proved in Section \ref{sec:C3_proof}.
    \begin{lemma}\label{lem:wasserstein_C3}
    The following holds for any two probability distributions $\nu$ and $\rho$ over $\bbR^n$:
    \begin{equation}
        \sup_{h \in \Lip(1)} \abs*{\bbE_{\Xvec\sim \nu}[h(\Xvec)] - \bbE_{\Yvec\sim\rho}[h(\Yvec)]} = \sup_{h \in \calH}\abs*{\bbE_{\Xvec\sim \nu}[h(\Xvec)] - \bbE_{\Yvec\sim\rho}[h(\Yvec)]}.
    \end{equation}
    \end{lemma}
    We now present the derivative bounds that we will need. The proof of Lemma \ref{lem:derivative_bounds} uses results from \cite{gorham2019measuring}, and is provided in \ref{sec:derivative_bounds_proof}.
    \begin{lemma}[Derivative Bounds]\label{lem:derivative_bounds}
        Consider the Stein Equation
        \begin{equation}
            G_{\Yvec} f_h(\xvec) = h(\xvec) - \bbE\left[h(\Yvec^\infty)\right],
        \end{equation}
        for $h\in \calH$. Assume that there exists a symmetric matrix $H>0$ and a real number $\phi>0$ such that
        \begin{equation}
            H\frac{\partial F(\vvec^*)}{\partial \vvec} + \left(\frac{\partial F(\vvec^*)}{\partial \vvec}\right)^TH \leq -2\phi H.
        \end{equation}
        Assume also that $S \diag\sqrt{\bar{\lavec}(\vvec^*)}$ has a right inverse. Then, there exists a solution $f_h$ to the Stein Equation, which satisfies
        \begin{equation}\label{eq:derivative_bound_first}
            \forall \xvec\in \bbR^n,\; \norm*{\grad f_h(\xvec)} \leq C_1,
        \end{equation}
        \begin{equation}\label{eq:derivative_bound_second}
            \forall \xvec\in \bbR^n,\; \norm*{\grad^2 f_h(\xvec)} \leq C_2,
        \end{equation}
        and for all $0 < \zeta < 1$, there exists $C_3(\zeta)$ such that
        \begin{equation}
            \forall \xvec,\xvec'\in \bbR^n,\; \norm*{\grad^2 f_h(\xvec) - \grad^2 f_h(\xvec')} \leq C_3(\zeta)\norm{\xvec-\xvec'}^{1-\zeta},
        \end{equation}
        where $C_1$, $C_2$, and $C_3$ are given by
        \begin{align}
            C_1 & = \int_0^\infty \norm*{e^{At}}dt,\\
            C_2 & = \norm*{\Sigma^{-1}}\cond (H) \left(2 + e^{-\phi}C_1\right),\\
            C_3 & = \frac{2C_3'}{\zeta} + e^{-\phi}C_3'C_1,
        \end{align}
        where $A = \frac{\partial F(\vvec^*)}{\partial \vvec}$, $\Sigma^{-1}$ is the right inverse of $S \diag\sqrt{\bar{\lavec}(\vvec^*)}$, and 
        \begin{equation}
            C_3' = 2 \max\{\norm*{\Sigma^{-1}},\norm*{\Sigma^{-1}}^2\}\cond (H)^{\tfrac{5}{2}}.
        \end{equation}
    \end{lemma}

    \begin{proof}
        See Section \ref{sec:derivative_bounds_proof}.
    \end{proof}

    \begin{remark}
        The assumption that there exists a symmetric matrix $H>0$ and a real number $\phi>0$ such that
        \begin{equation}
            H\frac{\partial F(\vvec^*)}{\partial \vvec} + \left(\frac{\partial F(\vvec^*)}{\partial \vvec}\right)^TH \leq -2\phi H
        \end{equation}
        is equivalent to the assumption that $\frac{\partial F(\vvec^*)}{\partial \vvec}$ is Hurwitz.
    \end{remark}
    
    Our goal is to analyze
    \begin{equation}
        \Wass\left(\tilde{\Xvec}_\Om^\infty, \Yvec^\infty\right) = \sup_{h\in\Lip(1)} \abs*{\bbE\left[h\left(\Xvectil_\Om^\infty\right)\right] - \bbE\left[h\left(\Yvec^\infty\right)\right]}.
    \end{equation}
    We begin by observing that per Lemma \ref{lem:wasserstein_C3} we can obtain $\Wass\left(\tilde{\Xvec}_\Om^\infty, \Yvec^\infty\right)$ by taking the supremum over $h\in \calH$ instead of $h\in\Lip(1)$, and thus we can analyze
    \begin{equation}
        \Wass\left(\tilde{\Xvec}_\Om^\infty, \Yvec^\infty\right) = \sup_{h\in\calH} \abs*{\bbE\left[h\left(\Xvectil_\Om^\infty\right)\right] - \bbE\left[h\left(\Yvec^\infty\right)\right]}.
    \end{equation}
    To do this we use Stein's Method following \cite{braverman2017stein}. The Stein Equation is
    \begin{equation}
        G_{\Yvec} f_h(\xvec) = h(\xvec) - \bbE\left[h(\Yvec^\infty)\right],
    \end{equation}
    which in our case is the Poisson equation
    \begin{equation}\label{eq:stein_equation}
        \frac{\partial f_h(\xvec)}{\partial \xvec}\cdot \frac{\partial F(\vvec^*)}{\partial \vvec}\xvec + \frac{1}{2}\sum_{i=1,j=1}^n D_{ij} \frac{\partial^2 f_h(\xvec)}{\partial x_i \partial x_j} = h(\xvec) - \bbE\left[h(\Yvec^\infty)\right],
    \end{equation}
    which is guaranteed to have a solution by Lemma \ref{lem:derivative_bounds}. Throughout the rest of this proof, we will denote by $f_h$ a solution to \eqref{eq:stein_equation} with $h\in \calH$ that satisfies the bounds given by Lemma \ref{lem:derivative_bounds}. We take the expectation of the Stein Equation with respect to $\Xvectil_\Om^\infty$ to find
    \begin{equation}\label{eq:stein_averaged}
        \bbE \left[ G_{\Yvec} f_h(\Xvectil_\Om^\infty)\right] = \bbE\left[ h(\Xvectil_\Om^\infty)\right] - \bbE\left[h(\Yvec^\infty)\right].
    \end{equation}
    Before proceeding, we show that $\bbE\left[ G_{\Xvectil_\Om}f_h(\Xvectil_\Om^\infty)\right] = 0$ for all $h\in \calH$. We have that $G_{\Xvectil_\Om}(\xvec,\xvec) = -\sum_{i=1}^r \lambda_i(\sqrt{\Om}\xvec + \Om \vvec^*)$, and so by Condition \ref{cond:k_1_moment_finite}, there exists $K>0$ such that
    \begin{equation}
        \sum_{\xvec} \mathbb{P}[\Xvectil_\Om^\infty = \xvec] \abs{G_{\Xvectil_\Om}(\xvec,\xvec) f_h(\xvec)} \leq \sum_{\xvec} \mathbb{P}[\Xvectil_\Om^\infty = \xvec] (K+\norm{\sqrt{\Om}\xvec + \Om \vvec^*}^\kappa) \abs{f_h(\xvec)},
    \end{equation}
    where the sums are taken over $\xvec$ such that $\sqrt{\Om}\xvec + \Om \vvec^* \in \calX_\Om$. By Lemma \ref{lem:derivative_bounds}, $f_h$ is Lipschitz, and so $\abs{f_h(\xvec)} \leq \abs{f_h(0)} + \norm{\xvec}$. Therefore,
    \begin{align}
         \sum_{\xvec} \mathbb{P}[\Xvectil_\Om^\infty = \xvec] \abs{G_{\Xvectil_\Om}(\xvec,\xvec) f_h(\xvec)} & \hspace{-1.5pt}\leq\hspace{-1.5pt} \sum_{\xvec} \mathbb{P}[\Xvectil_\Om^\infty = \xvec] (K\hspace{-1.5pt}+\hspace{-1.5pt}\norm{\sqrt{\Om}\xvec \hspace{-1.5pt}+\hspace{-1.5pt} \Om \vvec^*}^\kappa) \left(\abs{f_h(0)} \hspace{-1.5pt}+\hspace{-1.5pt} \norm{\xvec}\right)\hspace{-2pt},\\
         \begin{split}
         & \leq \sum_{\xvec} \mathbb{P}[\Xvectil_\Om^\infty = \xvec] \left((K\abs{f_h(0)} + K\norm{\xvec} \vphantom{\max\left\{1,\left(\norm{\sqrt{\Om}\xvec} + \norm{\Om \vvec^*}\right)^{\kappa}\right\}}\right.
         \\ & \left. \quad\quad\quad+ \abs{f_h(0)}\norm{\sqrt{\Om}\xvec + \Om \vvec^*}^\kappa \right. \\ &  \quad\quad\quad \left.+ \max\left\{1,\left(\norm{\sqrt{\Om}\xvec} + \norm{\Om \vvec^*}\right)^{\kappa}\right\} \norm{\xvec}\right),
        \end{split}
    \end{align}
    which is finite due to Condition \ref{cond:k_1_moment_finite}. Therefore, from Lemma \ref{lem:glynn}, $\bbE\left[ G_{\Xvectil_\Om}f_h(\Xvectil_\Om^\infty)\right] = 0$. Thus, for all $h \in \calH$, \eqref{eq:stein_averaged} implies that
    \begin{align}\label{eq:stein_discrep_bound}
        \abs*{\bbE\left[h\left(\Xvectil_\Om^\infty\right)\right] - \bbE\left[h\left(\Yvec^\infty\right)\right]} & = \abs*{\bbE \left[ G_{\Yvec} f_h(\Xvectil_\Om^\infty)\right] - \bbE \left[G_{\Xvectil_\Om}f_h(\Xvectil_\Om^\infty)\right]},\\
        & = \abs*{\bbE \left[ G_{\Yvec} f_h(\Xvectil_\Om^\infty) - G_{\Xvectil_\Om}f_h(\Xvectil_\Om^\infty)\right]},\\
        & \leq \bbE \abs*{ G_{\Yvec} f_h(\Xvectil_\Om^\infty) - G_{\Xvectil_\Om}f_h(\Xvectil_\Om^\infty)},
    \end{align}
    where $f_h$ is a solution to the Stein equation. The existence of $f_h$ is guaranteed by Lemma \ref{lem:derivative_bounds}. For conciseness, let us define $\delta = \frac{1}{\sqrt{\Om}}$. Let $\lvec = \xvec/\delta + \vvec^*/\delta^2$. For a fixed $\lvec$ we take the first order Taylor expansion of $G_{\Xvectil_\Om}f_h(\xvec)$ in $\delta$ using the Lagrange form of the remainder.
    \begin{align}
        G_{\Xvectil_\Om}f_h(\xvec) & = \sum_{i=1}^r \la_i(\lvec) \left(f(\xvec+\delta\xivec_i)-f(\xvec)\right),\\
        & = \delta \sum_{i=1}^r \la_i(\lvec) \frac{\partial f_h(\xvec)}{\partial \xvec} \xivec_i + \delta^2\frac{1}{2}  \sum_{i=1}^r \la_i(l) \xivec_i^T \grad^2{f_h}(\xvec + \ep \xivec_i) \xivec_i,\\
        \begin{split}\label{eq:taylor_expanded}
        & = \delta \sum_{i=1}^r \la_i(\lvec) \frac{\partial f_h(\xvec)}{\partial \xvec} \xivec_i + \delta^2\frac{1}{2}  \sum_{i=1}^r \la_i(l) \xivec_i^T \grad^2{f_h}(\xvec) \xivec_i\\ &\quad\quad + \delta^2\frac{1}{2}  \sum_{i=1}^r \la_i(\lvec) \left( \xivec_i^T \grad^2{f_h}(\xvec + \ep \xivec_i) \xivec_i - \xivec_i^T \grad^2{f_h}(\xvec) \xivec_i \right),
        \end{split}
    \end{align}
    where $0 \leq \ep \leq \delta$, and $\epsilon$ can depend on $\lvec$. We have that $\la_i(\lvec) = \frac{1}{\delta^2}\labar_i(\delta \xvec + \vvec^*)$, and so from Taylor's Theorem we have
    \begin{equation}
        \frac{1}{\delta^2}\labar_i(\delta \xvec + \vvec^*) = \frac{1}{\delta^2}\labar_i(\vvec^*) + \frac{1}{\delta}\frac{\partial \labar_i(\vvec^*)}{\partial \vvec} \xvec + R_i(\xvec),
    \end{equation}
    where $\abs{R_i(\xvec)} \leq \frac{1}{2}\norm{\xvec}^2\sup_{\zvec \in \calX}\norm{\grad^2{\labar_i}(\zvec)}$, in which the supremum over $\calX$ is finite due to our assumption that $\labar_i(\xvec)$ has bounded second derivative. We have that \eqref{eq:taylor_expanded} becomes
    \begin{align}
        \begin{split}
            G_{\Xvectil_\Om}f_h(\xvec) & = \sum_{i=1}^r \left(\frac{1}{\delta^2}\labar_i(\vvec^*) + \frac{1}{\delta}\frac{\partial \labar_i(\vvec^*)}{\partial \vvec} \xvec + R_i(\xvec)\right)\delta \frac{\partial f_h(\xvec)}{\partial \xvec} \xivec_i\\
            & + \sum_{i=1}^r \left(\frac{1}{\delta^2}\labar_i(\vvec^*) + \frac{1}{\delta}\frac{\partial \labar_i(\vvec^*)}{\partial \vvec} \xvec + R_i(\xvec)\right)\delta^2\frac{1}{2}  \xivec_i^T \grad^2{f_h}(\xvec) \xivec_i\\
            & + \sum_{i=1}^r \left(\frac{1}{\delta^2}\labar_i(\vvec^*) + \frac{1}{\delta}\frac{\partial \labar_i(\vvec^*)}{\partial \vvec} \xvec \right)\delta^2\frac{1}{2}  \left( \xivec_i^T \grad^2{f_h}(\xvec + \ep \xivec_i) \xivec_i - \xivec_i^T \grad^2{f_h}(\xvec) \xivec_i \right)\\
            & + \sum_{i=1}^r R_i(\xvec)\delta^2\frac{1}{2}  \left( \xivec_i^T \grad^2{f_h}(\xvec + \ep \xivec_i) \xivec_i - \xivec_i^T \grad^2{f_h}(\xvec) \xivec_i \right).
        \end{split}
    \end{align}
    Collecting terms and using the fact that $F(\vvec^*) = 0$, we obtain
    \begin{align}
        \begin{split}
        &G_{\Xvectil_\Om}f_h(\xvec)  = \frac{\partial f_h(\xvec)}{\partial \xvec}\sum_{i=1}^r  \xivec_i \frac{\partial \labar_i(\vvec^*)}{\partial \vvec} \xvec + \frac{1}{2}\sum_{i=1}^r\labar_i(\vvec^*)\xivec_i^T \grad^2{f_h}(\xvec) \xivec_i\\
        &\quad\quad + \delta \frac{1}{2}\sum_{i=1}^r \left[2R_i(\xvec) \frac{\partial f_h(\xvec)}{\partial \xvec} \xivec_i + \frac{\partial \labar_i(\vvec^*)}{\partial \vvec} \xvec  \xivec_i^T \grad^2{f_h}(\xvec) \xivec_i \right]\\
        &\quad\quad + \delta^2\frac{1}{2}\sum_{i=1}^r R_i(\xvec) \xivec_i^T \grad^2{f_h}(\xvec) \xivec_i\\
        &\quad\quad + \frac{1}{2}\sum_{i=1}^r \left(\labar_i(\vvec^*) + \delta\frac{\partial \labar_i(\vvec^*)}{\partial \vvec} \xvec + \delta^2 R_i(\xvec)\right) \left( \xivec_i^T \grad^2{f_h}(\xvec + \ep \xivec_i) \xivec_i - \xivec_i^T \grad^2{f_h}(\xvec) \xivec_i \right).
        \end{split}
    \end{align}
    Observe that the first two terms are identical to $G_Yf_h(\xvec)$, since
    \begin{equation}
        \frac{1}{2}\sum_{i=1}^r\labar_i(\vvec^*)\xivec_i^T \grad^2{f_h}(\xvec) \xivec_i = \frac{1}{2}\sum_{p=1,q=1}^n D_{pq}\frac{\partial f(\xvec)}{\partial x_p \partial x_q},
    \end{equation}
    with $D= S \diag(\lavecbar(\vvec^*))S^T$. Therefore, for all $\xvec$ such that $\sqrt{\Om}\xvec+\Om\vvec^*\in \calX_\Om$,
    \begin{align}
        \begin{split}
        &\abs*{G_{\Yvec} f_h(\xvec) - G_{\Xvectil_\Om}f_h(\xvec)}  \leq \abs*{\frac{1}{2}\delta \sum_{i=1}^r \left[2R_i(\xvec) \frac{\partial f_h(\xvec)}{\partial \xvec} \xivec_i + \frac{\partial \labar_i(\vvec^*)}{\partial \vvec} \xvec  \xivec_i^T \grad^2{f_h}(\xvec) \xivec_i \right]}\\
         & \quad\quad\quad + \abs*{\frac{1}{2}\delta^2\sum_{i=1}^r R_i(\xvec) \xivec_i^T \grad^2{f_h}(\xvec) \xivec_i}\\
         & \quad\quad + \abs*{\frac{1}{2}\sum_{i=1}^r \left(\labar_i(\vvec^*) + \delta\frac{\partial \labar_i(\vvec^*)}{\partial \vvec} \xvec + \delta^2R_i(\xvec)\right) \left( \xivec_i^T \grad^2{f_h}(\xvec + \ep \xivec_i) \xivec_i - \xivec_i^T \grad^2{f_h}(\xvec) \xivec_i \right)}.
         \end{split}
    \end{align}
    Letting $Q_i = \sup_{\vvec\in\calX} \norm{\grad^2{\labar_i}(\vvec)}$ and invoking Lemma \ref{lem:derivative_bounds} by the fact that Assumptions \ref{assum:propensities} and \ref{assum:stoich_full_dim} ensure that $S \diag \sqrt{\lavecbar(\vvec^*)}$ has a right inverse since $\lavecbar(\vvec^*)>0$ and $S$ has full column rank, we have
    \begin{align}
         \begin{split}
         \abs*{G_{\Yvec} f_h(\xvec) - G_{\Xvectil_\Om}f_h(\xvec)} & \leq \frac{1}{2}\delta \sum_{i=1}^r \left(C_1 Q_i \norm*{\xvec}^2\norm{\xivec_i} + C_2\norm*{\frac{\partial \labar_i(\vvec^*)}{\partial \vvec}}\norm*{\xvec}\norm{\xivec_i}^2\right) \\
         & + \frac{1}{4}\delta^2  C_2\norm*{\xvec}^2 \sum_{i=1}^r Q_i \norm*{\xivec_i}^2\\
         & + \frac{1}{2}\delta^{2-\zeta}  C_3(\zeta) \norm{\xvec} \sum_{i=1}^r \norm*{\frac{\partial \labar_i(\vvec^*)}{\partial \vvec}}\norm*{\xivec_i}^{3-\zeta}\\
         & + \frac{1}{4}\delta^{3-\zeta}  C_3(\zeta) \norm*{\xvec}^2 \sum_{i=1}^r Q_i \norm*{\xivec_i}^{3-\zeta}\\
         & + \frac{1}{2}\delta^{1-\zeta} C_3(\zeta) \sum_{i=1}^r \abs*{\labar_i(\vvec^*)} \norm*{\xivec_i}^{3-\zeta}.
        \end{split}
    \end{align}
    This bound holds for any $h \in \calH$, with $f_h$ chosen as the solution to the Stein equation \eqref{eq:stein_equation} specified by Lemma \ref{lem:derivative_bounds}, and thus from \eqref{eq:stein_discrep_bound} we have
    \begin{equation}
        \begin{split}
        \abs*{\bbE\left[h\left(\Xvectil_\Om^\infty\right)\right] - \bbE\left[h\left(\Yvec^\infty\right)\right]}& \\
        & \hspace{-1.5in} \leq \bbE\left[\frac{1}{2}\delta \sum_{i=1}^r \left(C_1 Q_i \norm*{\Xvectil_\Om^\infty}^2\norm{\xivec_i} + C_2\norm*{\frac{\partial \labar_i(\vvec^*)}{\partial \vvec}}\norm*{\Xvectil_\Om^\infty}\norm{\xivec_i}^2\right) \right.\\
        + \frac{1}{4}\delta^2  C_2\norm*{\Xvectil_\Om^\infty}^2 \sum_{i=1}^r Q_i \norm*{\xivec_i}^2\\
        + \frac{1}{2}\delta^{2-\zeta}  C_3(\zeta) \norm{\Xvectil_\Om^\infty} \sum_{i=1}^r \norm*{\frac{\partial \labar_i(\vvec^*)}{\partial \vvec}}\norm*{\xivec_i}^{3-\zeta}\\
        + \frac{1}{4}\delta^{3-\zeta}  C_3(\zeta) \norm*{\Xvectil_\Om^\infty}^2 \sum_{i=1}^r Q_i \norm*{\xivec_i}^{3-\zeta}\\
        \left.+ \frac{1}{2}\delta^{1-\zeta} C_3(\zeta) \sum_{i=1}^r \abs*{\labar_i(\vvec^*)} \norm*{\xivec_i}^{3-\zeta}\right].
        \end{split}
    \end{equation}
    Using the moment bounds and the linearity of the expectation we have
    \begin{equation}
        \begin{split}
        \abs*{\bbE\left[h\left(\Xvectil_\Om^\infty\right)\right] - \bbE\left[h\left(\Yvec^\infty\right)\right]} \leq \frac{1}{2}\delta \sum_{i=1}^r \left(C_1 Q_i M_2 \norm{\xivec_i} + C_2\norm*{\frac{\partial \labar_i(\vvec^*)}{\partial \vvec}}M_1\norm{\xivec_i}^2\right) \\
        + \frac{1}{4}\delta^2  C_2 M_2 \sum_{i=1}^r Q_i \norm*{\xivec_i}^2\\
        + \frac{1}{2}\delta^{2-\zeta}  C_3(\zeta) M_1 \sum_{i=1}^r \norm*{\frac{\partial \labar_i(\vvec^*)}{\partial \vvec}}\norm*{\xivec_i}^{3-\zeta}\\
        + \frac{1}{4}\delta^{3-\zeta}  C_3(\zeta) M_2 \sum_{i=1}^r Q_i \norm*{\xivec_i}^{3-\zeta}\\
        + \frac{1}{2}\delta^{1-\zeta} C_3(\zeta) \sum_{i=1}^r \abs*{\labar_i(\vvec^*)} \norm*{\xivec_i}^{3-\zeta}.
        \end{split}
    \end{equation}
    To prove the error bound proportional to $\frac{\ln \Omega}{\sqrt{\Omega}}$, we take $\zeta = -(\ln \delta )^{-1}$ for $\Om > e^2$, in which case $C_3(\zeta) = -2\ln(\delta) C_3' + e^{-\phi}C_3'C_1$. Let $\iota > 0$. For all $\Om \geq \max\{\Om_0, e^2+\iota\}$ we have
    \begin{equation}
        \begin{split}
        \abs*{\bbE\left[h\left(\Xvectil_\Om^\infty\right)\right] - \bbE\left[h\left(\Yvec^\infty\right)\right]} \leq \frac{1}{2}\delta \sum_{i=1}^r \left(C_1 Q_i M_2\norm{\xivec_i} + C_2\norm*{\frac{\partial \labar_i(\vvec^*)}{\partial \vvec}}M_1\norm{\xivec_i}^2\right) \\
        + \frac{1}{4} \delta^2 C_2 M_2 \sum_{i=1}^r Q_i \norm*{\xivec_i}^2\\
        + \frac{1}{2}\delta^2 C_3'e\left(2\ln (\delta^{-1}) + e^{-\phi}C_1\right) M_1 \sum_{i=1}^r \norm*{\frac{\partial \labar_i(\vvec^*)}{\partial \vvec}}\max\{\norm*{\xivec_i}^2, \norm*{\xivec_i}^{3}\}\\
        + \frac{1}{4} \delta^{3} C_3'e \left(2\ln (\delta^{-1}) + e^{-\phi}C_1\right) M_2 \sum_{i=1}^r Q_i \max\{\norm*{\xivec_i}^2, \norm*{\xivec_i}^{3}\}\\
        + \frac{1}{2} \delta C_3'e \left(2\ln (\delta^{-1}) + e^{-\phi}C_1\right) \sum_{i=1}^r \abs*{\labar_i(\vvec^*)} \max\{\norm*{\xivec_i}^2, \norm*{\xivec_i}^{3}\}.
        \end{split}
    \end{equation}
    The $\delta \ln (\delta^{-1}) = \frac{1}{2}\delta \ln (\delta^{-2}) = \frac{1}{2}\frac{\ln \Omega}{\sqrt{\Omega}}$ term decays the slowest, which yields the desired result.

\subsection{Proof of Lemma \ref{lem:derivative_bounds}}\label{sec:derivative_bounds_proof}
    Here we prove Lemma \ref{lem:derivative_bounds}. The proof is an application of Theorem 5 in \cite{gorham2019measuring}, with the bounds tightened with more detailed analysis in our special case. For convenience we adopt the following notion from \cite{gorham2019measuring}. For a function $g$ with domain $\bbR^n$, we use the notation
    \begin{equation}
        M_j(g) = \sup_{\xvec,\xvec' \in \bbR^n, \xvec \neq \xvec'} \frac{\norm{\nabla^{\lceil j\rceil-1}g(\xvec) - \nabla^{\lceil j\rceil-1}g(\xvec')}}{\norm{\xvec-\xvec'}^{\{j\}}},
    \end{equation}
    where $\{j\} = j - \lceil j - 1\rceil$, and
    \begin{equation}
        F_j(g) = \sup_{\xvec \in \bbR^n,\norm{v_1}=\norm{v_2}=\dots = \norm{v_j}=1} \norm{\nabla^jg(\xvec)[\zvec_1,\dots,\zvec_j]}_F
    \end{equation}
    Here, for $g:\bbR^n\rightarrow \bbR$, $\nabla^r g(\xvec)[\zvec_1,\dots,\zvec_j]$ denotes the $r$\textsuperscript{th} derivative of $g$ along the directions $\zvec_1,\dots,\zvec_j$. An It\^{o} Diffusion $\Zvec(t)$ is the solution to
    \begin{equation}\label{eq:ito_diffusion}
            d\Zvec_{\xvec}(t) = b(\Zvec_{\xvec}(t))dt + \sigma(\Zvec_{\xvec}(t))d\Wvec(t), \; \Zvec_{\xvec}(t) = \xvec,
    \end{equation}
    where $\Wvec(t)$ is an $m$ dimensional Wiener process. We denote the transition semigroup of $\Zvec(t)$ by $(P_t)_{t\geq 0}$, and the infinitesimal generator of $\Zvec(t)$ by $\calA$. $\Zvec(t)$ has \emph{Wasserstein decay rate r(t)} for a nonincreasing integrable function $r:\bbR_{\geq 0}\rightarrow \bbR$ if for all $\zvec,\zvec'\in\bbR^n$ and all $t\geq 0$,
    \begin{equation}
        \Wass\left(\Zvec(t)|\Zvec(0)=\zvec,\Zvec(t)|\Zvec(0)=\zvec'\right) \leq r(t)\Wass(\zvec,\zvec').
    \end{equation}
     For clarity, we start by stating the following theorem, slightly adapted from \cite{gorham2019measuring}, which ensures the existence of constants $C_1$, $C_2$ and $C_3'$. We will shortly improve these constants in our special case.
    \begin{theorem}[Theorem 5 in \cite{gorham2019measuring}]\label{thm:gorham5}
        Let $h:\bbR^n\rightarrow\bbR$ be Lipschitz. Consider an $n$ dimensional It\^{o} diffusion given by \eqref{eq:ito_diffusion} with Wasserstein decay rate $\bar{r}(t)$ and invariant measure $P$. If $b$ and $\sigma$ have locally Lipschitz second and third derivatives and a right inverse $\sigma^{-1}(\xvec)$ for each $\xvec \in \bbR^n$, and $h \in C^3(\bbR^n)$ with bounded second and third derivatives, then
        \begin{equation}
            f_h  = \int_0^\infty \bbE_P\left[h(\Zvec)\right] - P_t h dt
        \end{equation}
        is twice continuously differentiable, satisfies the Stein equation
        \begin{equation}
            \calA f_h = h - \bbE_P\left[h(\Zvec)\right],
        \end{equation}
        and additionally,
        \begin{equation}
            M_1(f_h) \leq M_1(h) \int_0^\infty \bar{r}(t)dt,
        \end{equation}
        and
        \begin{equation}
            M_2(f_h) \leq M_1(h)(\beta_1 + \beta_2),
        \end{equation}
        where
        \begin{align}
            \beta_1 = \bar{r}(0)\left(2M_0(\sigma^{-1}) + \bar{r}(0)M_1(\sigma)M_0(\sigma^{-1}) + \bar{r}(0) \sqrt{\alpha} \right),\\
            \beta_2 = \bar{r}(0)\left( e^{\ga_2}M_0(\sigma^{-1}) + e^{\ga_2}M_1(\sigma)M_0(\sigma^{-1}) + \frac{2}{3}e^{\ga_4}\sqrt{\al} \right),
        \end{align}
        with $\ga_{\rho} = \rho M_1(b) + \frac{\rho^2-2\rho}{2}M_1(\sigma)^2 + \frac{\rho}{2}F_1(\sigma)^2$ and $\al = \frac{M_2(b)^2}{2M_1(b)+4M_1(\sigma)^2} + 2F_2(\sigma)^2$. Furthermore, for all $\zeta \in (0,1)$,
        \begin{equation}
            M_{3-\zeta}(f_h) \leq M_1(h) \frac{1}{K}\left(\frac{1}{\zeta} + \int_0^\infty \bar{r}(t) dt \right),
        \end{equation}
        for $K>0$ that depends only on $M_{1:3}(\sigma)$, $M_{1:3}(b)$, $M_0(\sigma^{-1})$, and $\bar{r}$.
    \end{theorem}
    
    We now show that $\Yvec(t)$ has a Wasserstein decay rate of $\norm*{e^{At}}$, meaning that
    \begin{equation}
        \Wass(\Yvec(t)|\Yvec(0) = \yvec, \Yvec(t) | \Yvec(0) = \yvec') \leq \norm*{e^{At}} \Wass(\yvec,\yvec'),
    \end{equation}
    where $\Yvec(t)|\Yvec(0) = \yvec$ denotes $\Yvec(t)$, the solution to \eqref{eq:LNA_FLUC} with $\vvec(t) = \vvec^*$ and $\Yvec'(0) = \yvec$. Observe that $\Yvec(t)|\Yvec(0) = \yvec$ and $\Yvec(t)|\Yvec(0) = \yvec'$ are Gaussian with the same covariance. Therefore, letting $\Yvec_{\yvec} = \Yvec(t)|\Yvec(0) = \yvec$, we have
    \begin{equation}\label{eq:Wasserstein_Gaussian_same_covariance}
        \Wass\left(\Yvec_{\yvec}(t), \Yvec_{\yvec'}(t)\right) = \norm*{\bbE\left[\Yvec_{\yvec}(t)\right] - \bbE\left[\Yvec_{\yvec'}(t)\right]},
    \end{equation}
    which can be shown by considering the lower bound 
    \begin{equation}
        \Wass\left(\Yvec_{\yvec}(t), \Yvec_{\yvec'}(t)\right) \geq \norm*{\bbE\left[\Yvec_{\yvec}(t)\right] - \bbE\left[\Yvec_{\yvec'}(t)\right]}
    \end{equation}
    from Jensen's inequality, and the upper bound
    \begin{equation}
        \Wass\left(\Yvec_{\yvec}(t), \Yvec_{\yvec'}(t)\right) \leq \norm*{\bbE\left[\Yvec_{\yvec}(t)\right] - \bbE\left[\Yvec_{\yvec'}(t)\right]}
    \end{equation}
    as evidenced by the coupling $\Yvec_{\yvec}(t) = \Yvec_{\yvec'}(t) + \bbE\left[\Yvec_{\yvec}(t)\right] - \bbE\left[\Yvec_{\yvec'}(t)\right]$. Theorem \ref{thm:gorham5} then gives us the constant $C_1$. To prove the claimed expressions for $C_2$ and $C_3'$, we must specialize the arguments of \cite{gorham2019measuring} to our diffusion, which has a linear and Hurwitz $b$ and a uniform $\sigma$. As in \cite{gorham2019measuring}, denote by $\Vvec_{\vvec}(t)$ the first \emph{directional derivative flow}, defined as the solution to
    \begin{equation}\label{eq:first_dir_flow}
        d\Vvec_{\vvec}(t) = \grad b(\Zvec_{\xvec}(t))\Vvec_{\vvec}(t) dt + \grad \sigma(\Zvec_{\xvec}(t))d\Wvec(t), \quad \Vvec_{\vvec}(0) = \vvec.
    \end{equation}
    Further, define $\Uvec_{\vvec,\vvec'}(t)$, the \emph{second directional derivative flow}, defined as the solution to
    \begin{equation}\label{eq:second_dir_flow}
        \begin{split}
        d\Uvec_{\vvec,\vvec'}(t) = \left(\grad b(\Zvec_{\xvec}(t))\Uvec_{\vvec,\vvec'} + \grad^2 b(\Zvec_{\xvec}(t))[\Vvec_{\vvec'}(t)]\Vvec_{\vvec}(t)\right)dt + \left(\grad \sigma(\Zvec_{\xvec}(t))\Uvec_{\vvec,\vvec'}(t)\right.\\ \left.+ \grad^2\sigma(\Zvec_{\xvec}(t))[\Vvec_{\vvec'}(t)]\Vvec_{\vvec}(t)\right)d\Wvec(t),\; \Uvec_{\vvec,\vvec'}(0) = 0.
        \end{split}
    \end{equation}
    In our proof the following lemma replaces the derivative flow bounds given in Lemma 16 of \cite{gorham2019measuring}.
    \begin{lemma}\label{lem:linear_derivative_flow_bounds}
        Consider an It\^{o} diffusion given by \eqref{eq:ito_diffusion} with transition semigroup $(P_t)_{t\geq 0}$ and assume that $b(\zvec) = A\zvec$ where $A \in \bbR^{n\times n}$ satisfies
        \begin{equation}
            HA+A^TH \leq -2\phi H,
        \end{equation}
        For a matrix $H >0$ and a real number $\phi>0$. Assume also that $\sigma(\zvec) = \sigma$ is a constant. Then, for all $t \geq 0$, $\rho > 0$, and $\vvec \in \bbR^n$,
        \begin{equation}
            \bbE\left[\norm{\Vvec_{\vvec}(t)}^\rho\right] \leq \mathrm{cond}(H)^{\rho/2}e^{-\rho\phi t}\norm{\vvec}^{\rho}
        \end{equation}
    and for all $t \geq 0$, $\Uvec_{\vvec,\vvec'}(t) = 0$.
    \end{lemma}
    \begin{proof}
        Under the assumptions of the Lemma, \eqref{eq:first_dir_flow} becomes
        \begin{equation}
            d\Vvec_{\vvec}(t) = \frac{\partial F(\vvec^*)}{\partial \vvec}\Vvec_{\vvec}(t) dt , \quad \Vvec_{\vvec}(0) = \vvec,
        \end{equation}
    and thus $\Vvec_{\vvec}(t) = e^{At}\vvec$. Therefore, $\bbE\left[\Vvec_{\vvec}(t)^\rho\right] = \norm{e^{At}\vvec}^\rho$. To bound this quantity, observe that the condition
        \begin{equation}
            HA+A^TH \leq -2\phi H,
        \end{equation}
        implies that
        \begin{equation}
            \sqrt{H}A\sqrt{H}^{-1}+\left(\sqrt{H}A\sqrt{H}^{-1}\right)^T \leq -2\phi I,
        \end{equation}
        where $\sqrt{H}$ is the principle square root of $H$. It then follows from standard arguments from linear systems theory \cite{desoer1972measure,lohmiller2000nonlinear} that for all $t\geq 0$,
        \begin{equation}
             \norm{e^{\sqrt{H}A\sqrt{H}^{-1} t}}  \leq e^{-\phi t}.
        \end{equation}
        We then have that for all $t\geq 0$ and $\vvec\in\bbR^n$,
        \begin{align}
             \norm{e^{A t}\vvec} & = \norm{\sqrt{H}^{-1}e^{\sqrt{H}A\sqrt{H}^{-1} t}\sqrt{H}\vvec} \\
            & \leq \norm{\sqrt{H}^{-1}}\norm{e^{\sqrt{H}A\sqrt{H}^{-1} t}}\norm{\sqrt{H}}\norm{\vvec}\\
            & \leq \sqrt{\mathrm{cond}(H)}e^{-\phi t}\norm{\vvec}
        \end{align}
        which establishes the desired bound on $\norm{e^{At}\vvec}$. To complete the proof, observe that under the assumptions of the lemma, \eqref{eq:second_dir_flow} becomes
    \begin{equation}
        d\Uvec_{\vvec,\vvec'}(t) = \grad b(\Zvec_{\xvec}(t))\Uvec_{\vvec,\vvec'},\; \Uvec_{\vvec,\vvec'}(0) = 0,
    \end{equation}
    and thus $\Uvec_{\vvec,\vvec'}(t) = 0$.
    \end{proof}
    With this specialization, we now trace the construction of $C_2$ and $C_3(\zeta)$ through the same argument used in proof of Theorem 5 in \cite{gorham2019measuring} to obtain the claimed expressions. Using Lemma \ref{lem:linear_derivative_flow_bounds} we make the specializations described in Lemmas \ref{lem:M2_bound} and \ref{lem:M3_bound}, which we state and prove before proceeding. In our case, the bound on $M_2(f)$ given by Lemma 15 in \cite{gorham2019measuring} can be replaced by the following lemma.
    \begin{lemma}\label{lem:M2_bound}
        Consider an It\^{o} diffusion given by \eqref{eq:ito_diffusion} with transition semigroup $(P_t)_{t\geq 0}$ and assume that $b(\zvec) = A\zvec$ where $A \in \bbR^{n\times n}$ satisfies
        \begin{equation}
            HA+A^TH \leq -2\phi H,
        \end{equation}
        For a matrix $H >0$ and a real number $\phi>0$. Assume also that $\sigma(\zvec) = \sigma$ is a constant with right inverse $\sigma^{-1}$. Then, for all $t>0$ and $f\in C^2$ with bounded first and second derivatives, $P_tf$ satisfies
        \begin{equation}
            M_2(P_tf) \leq \inf_{t_0 \in (0,t]} M_1(f) r(t-t_0) M_0(\sigma^{-1})e^{-\phi t_0}\mathrm{cond}(H)\frac{1}{\sqrt{t_0}},
        \end{equation}
        where $r(t) = \norm{e^{At}}$.
    \end{lemma}
    \begin{proof}
        First, observe that under the assumptions of the lemma, the It\^{o} diffusion has a Wasserstein decay rate of $r(t) = \norm{e^{At}}$. Thus, we can invoke Lemma 15 directly to obtain that $P_tf$ is twice continuously differentiable. The remainder of the proof follows that of Lemma 15 in \cite{gorham2019measuring}, with the substitution of our form of $b$ and $\sigma$, and using Lemma \ref{lem:linear_derivative_flow_bounds} instead of Lemma 16 in \cite{gorham2019measuring} for bounds on the derivative flows.
    \end{proof}
    In our special case, the bound on $M_3(P_tf)$ from Lemma 20 in \cite{gorham2019measuring} can be replaced with the following lemma.
    \begin{lemma}\label{lem:M3_bound}
        Consider an It\^{o} diffusion given by \eqref{eq:ito_diffusion} with transition semigroup $(P_t)_{t\geq 0}$ and assume that $b(\zvec) = A\zvec$ where $A \in \bbR^{n\times n}$ satisfies
        \begin{equation}
            HA+A^TH \leq -2\phi H,
        \end{equation}
        For a matrix $H >0$ and a real number $\phi>0$. Assume also that $\sigma(\zvec) = \sigma$ is a constant with right inverse $\sigma^{-1}$. Then, for all $t>0$ and $f\in C^3$ with bounded second and third derivatives, $P_tf$ satisfies
        \begin{equation}
            M_3(P_t f) \leq \inf_{t_0\in(0,t]} \frac{2}{t_0}M_1(f)r(t-t_0)M_0^2(\sigma^{-1})\mathrm{cond}(H)^{\frac{5}{2}}e^{-\frac{3}{2}\phi t_0},
        \end{equation}
        where $r(t) = \norm{e^{At}}$.
    \end{lemma}
    \begin{proof}
        First, observe that under the assumptions of the lemma, the It\^{o} diffusion has a Wasserstein decay rate of $r(t) = \norm{e^{At}}$. Then, the proof follows that of Lemma 20 in \cite{gorham2019measuring}, but with the substitution of the bound on $M_2(P_tf)$ given by Lemma \ref{lem:M2_bound}, and the derivative flow bounds of Lemma 16 in \cite{gorham2019measuring} replaced by Lemma \ref{lem:linear_derivative_flow_bounds}.
    \end{proof}
    We now finish our proof of Lemma \ref{lem:derivative_bounds}. Observe that in this case we have $A = \frac{\partial F}{\partial \xvec}(\xvec^*)$ and $\sigma = \Sigma$, where $\Sigma$ has a right inverse $\Sigma^{-1}$. Let $0 < \zeta < 1$. We have for $f_h$ solving $\calA f_h = h - \bbE_P\left[h(\Zvec)\right]$ that by the Dominated Convergence Theorem and Jensen's inequality that
    \begin{equation}
        M_{\zeta-1}(\grad ^2 f_h) = M_{\zeta-1}\left(-\int_0^\infty \grad^2 P_th dt\right) \leq \int_0^\infty M_{1-\zeta}\left(\grad^2 P_t h\right)dt.
    \end{equation}
    Splitting the last integral into the interval $[0,1]$ and $[1,\infty]$ we have
    \begin{equation}
        M_{\zeta-1}(\grad ^2 f_h) \leq T_1 + T_2,
    \end{equation}
    where
    \begin{equation}
        T_1 = \int_0^1 M_{1-\zeta}\left(\grad^2 P_t h\right)dt,
    \end{equation}
    and
    \begin{equation}
        T_2 = \int_1^\infty M_{1-\zeta}\left(\grad^2 P_t h\right)dt.
    \end{equation}
    Applying the seminorm interpolation result of Lemma 19 in \cite{gorham2019measuring}, we have
    \begin{align}
        M_{1-\zeta}\left(\grad^2 P_t h\right) & \leq   2^\zeta \left( M_0(\grad^2P_t h) \right)^\zeta \left( M_1(\grad^2P_t h) \right)^{1-\zeta}\\
        & = 2^\zeta \left( M_2(P_t h) \right)^\zeta \left( M_3(P_t h) \right)^{1-\zeta}.
    \end{align}
    Under the assumptions of the lemma, we can apply Lemmas \ref{lem:M2_bound} and \ref{lem:M3_bound}. For $t \leq 1$ we set $t_0 = t$ to obtain a bound on $T_1$ as follows.
    \begin{align}
        T_1 & = \int_0^1 M_{1-\zeta}\left(\grad^2 P_t h\right)dt\\
        & \leq \int_0^1 2^\zeta \left( M_2(P_t h) \right)^\zeta \left( M_3(P_t h) \right)^{1-\zeta} dt\\
        & \leq \int_0^1 2^\zeta \left( r(0)M_0(\Sigma^{-1})e^{-\phi t} \mathrm{cond}(H) \frac{1}{\sqrt{t}} \right)^\zeta \left( \frac{2}{t}r(0)M_0^2(\Sigma^{-1})\mathrm{cond}(H)^{\frac{5}{2}}e^{-\frac{3}{2}\phi t} \right)^{1-\zeta} dt\\
        & = 2 r(0)M_0^{2-\zeta}(\Sigma^{-1})\mathrm{cond}(H)^{\frac{5}{2}-\frac{3}{2}\zeta} \int_0^1 e^{-(\frac{3}{2}-\frac{1}{2}\zeta)\phi t} t^{-1 + \frac{1}{2}\zeta} dt\\
        & \leq 2 r(0)M_0^{2-\zeta}(\Sigma^{-1})\mathrm{cond}(H)^{\frac{5}{2}-\frac{3}{2}\zeta}\frac{2}{\zeta}\\
        & = \frac{4}{\zeta} \max\left\{ \norm{\Sigma^{-1}}, \norm{\Sigma^{-1}}^2 \right\} \mathrm{cond}(H)^{\frac{5}{2}},
    \end{align}
    where we have used the observations that $r(0) = 1$ and $\mathrm{cond}(H) \geq 1$. turning to $T_2$, for $t\geq 1$ we set $t_0 = 1$ in Lemmas \ref{lem:M2_bound} and \ref{lem:M3_bound}, obtaining
    \begin{align}
        T_2 & = \int_1^\infty M_{1-\zeta}\left(\grad^2 P_t h\right)dt\\
        & \leq \int_1^\infty 2^\zeta \left( M_2(P_t h) \right)^\zeta \left( M_3(P_t h) \right)^{1-\zeta} dt\\
        & \leq \int_1^\infty \left( 2 r(t-1)M_0(\Sigma^{-1})e^{-\phi} \mathrm{cond}(H) \right)^\zeta \left( 2r(t-1)M_0^2(\Sigma^{-1})\mathrm{cond}(H)^{\frac{5}{2}}e^{-\frac{3}{2}\phi} \right)^{1-\zeta} dt\\
        & = 2 M_0^{2-\zeta}(\Sigma^{-1})\mathrm{cond}(H)^{\frac{5}{2}-\frac{3}{2}\zeta} e^{-(\frac{3}{2}-\frac{1}{2}\zeta)\phi} \int_1^\infty r(t-1) dt\\
        & \leq 2 M_0^{2-\zeta}(\Sigma^{-1})\mathrm{cond}(H)^{\frac{5}{2}-\frac{3}{2}\zeta} e^{-\phi} \int_0^\infty r(t) dt\\
        & = 2 \max\left\{ \norm{\Sigma^{-1}}, \norm{\Sigma^{-1}}^2 \right\} \mathrm{cond}(H)^{\frac{5}{2}} e^{-\phi} \int_0^\infty r(t) dt.
    \end{align}
    Finally, by the definition of $M_{1-\zeta}$ we have 
    \begin{equation}
        \forall \xvec,\xvec' \in \bbR^n, \; \norm{\grad^2f_h(\xvec)-\grad^2f_h(\xvec')} \leq (T_1+T_2)\norm{\xvec-\xvec'}^{1-\zeta},
    \end{equation}
    which proves the desired bound when combined with the bounds on $T_1$ and $T_2$ derived above.

\subsection{Proof of Lemma \ref{lem:wasserstein_C3}}\label{sec:C3_proof}
\begin{proof}
    We have that 
    \begin{equation}
        \sup_{h \in \Lip(1)} \abs*{\bbE_{\Xvec\sim \nu}[h(\Xvec)] - \bbE_{\Yvec\sim\rho}[h(\Yvec)]} \geq \sup_{h \in \calH}\abs*{\bbE_{\Xvec\sim \nu}[h(\Xvec)] - \bbE_{\Yvec\sim\rho}[h(\Yvec)]}.
    \end{equation}
    We will show that 
    \begin{equation}
        \sup_{h \in \Lip(1)} \abs*{\bbE_{\Xvec\sim \nu}[h(\Xvec)] - \bbE_{\Yvec\sim\rho}[h(\Yvec)]} \leq \sup_{h \in \calH}\abs*{\bbE_{\Xvec\sim \nu}[h(\Xvec)] - \bbE_{\Yvec\sim\rho}[h(\Yvec)]}.
    \end{equation}
    Our proof is based on mollifying $h\in\Lip(1)$. Let $\varphi_\ep(\xvec)$ be the bump function defined by
    \begin{equation}
        \varphi_\ep (\xvec) = \left\{\begin{array}{cc}
             \frac{1}{\ep^n}Ce^{\frac{-1}{1 - \norm{\xvec/\ep}^2}}, & \norm{\xvec/\ep}\leq 1, \\
             0, & \text{else},
        \end{array}\right.
    \end{equation}
    Where $C = \left(\int_{\bbR^n} e^{\frac{-1}{1 - \norm{\xvec}^2}}d\xvec\right)^{-1}$. Suppose for contradiction that 
    \begin{equation}
        \sup_{h \in \Lip(1)} \abs*{\bbE_{\Xvec\sim \nu}[h(\Xvec)] - \bbE_{\Yvec\sim\rho}[h(\Yvec)]} > \sup_{h \in \calH}\abs*{\bbE_{\Xvec\sim \nu}[h(\Xvec)] - \bbE_{\Yvec\sim\rho}[h(\Yvec)]},
    \end{equation}
    then, there exists $h^*\in \Lip(1)$ such that $h^* \notin \calH$, and
    \begin{equation}\label{eq:Lip_C3_gap}
        \abs*{\bbE_{\Xvec\sim \nu}[h^*(\Xvec)] - \bbE_{\Yvec\sim\rho}[h^*(\Yvec)]} - \sup_{h \in \calH}\abs*{\bbE_{\Xvec\sim \nu}[h(\Xvec)] - \bbE_{\Yvec\sim\rho}[h(\Yvec)]} =: a > 0.
    \end{equation}
    Let $h^*_{a/4} = h^* \star \varphi_{a/4}$. We will show that $h^*_{a/4} = h^* \star \varphi_{a/4} \in \calH$, and $\norm{h^* - h^*_{a/4}}_\infty \leq a/4$. Let $\ep = a/4$. Recall that 
    \begin{equation}
        h^*_{\ep}(\xvec) = \left(\varphi_\ep\star h^*\right)(\xvec) = \int_{\bbR^n} \varphi_\ep(\xvec-\yvec)h(\yvec)\mu(d\yvec),
    \end{equation}
    where $\mu$ is the Lebesgue measure. Observe that
    \begin{align}
        \abs*{h^*_{\ep}(\xvec) - h^*_{\ep}(\xvec')} & = \abs*{\int_{\bbR^n} \varphi_\ep(\yvec)\left(h(\xvec-\yvec)-h(\xvec'-\yvec)\right)\mu(d\yvec)},\\
        & \leq \int_{\bbR^n} \abs*{\varphi_\ep(\yvec)\left(h(\xvec-\yvec)-h(\xvec'-\yvec)\right)}\mu(d\yvec),\\
        & \leq \int_{\bbR^n} \varphi_\ep(\yvec)\abs*{h(\xvec-\yvec)-h(\xvec'-\yvec)}\mu(d\yvec),\\
        & \leq \norm{\xvec-\xvec'} \int_{\bbR^n} \varphi_\ep(\yvec)\mu(d\yvec),\\
        & = \norm{\xvec-\xvec'},
    \end{align}
    and thus $h^*_\ep \in \Lip(1)$. Next, let $\frac{\partial^{\alvec}}{\partial \xvec^{\alvec}}$ with $\alvec \in \bbZ_{\geq 0}^n$ denote $\prod_{j=1}^n \frac{\partial^{\al_j}}{\partial x_j^{\al_j}}$, and observe that by the Dominated Convergence Theorem,
    \begin{equation}
        \frac{\partial^{\alvec}h^*_\ep(\xvec)}{\partial \xvec^{\alvec}}  = \int_{\bbR^n} \frac{\partial^{\alvec} \varphi_\ep(\xvec-\yvec)}{\partial \xvec^{\alvec}}h(\yvec) \mu(d\yvec),
    \end{equation}
    showing that $h^*_\ep(\xvec) \in C^\infty(\bbR^n)$, and that furthermore that $\norm*{\frac{\partial^{\alvec}h^*_\ep(\xvec)}{\partial \xvec^{\alvec}}}$ is bounded. We have shown that $h^*_\ep \in \calH$. Now, consider
    \begin{align}
        \abs*{h^*(\xvec) - h^*_{a/4}(\xvec)} & = \abs*{\int_{\bbR^n}\varphi_{a/4}(\yvec)h^*(\xvec) \mu(d\yvec) - \int_{\bbR^n} \varphi_{a/4}(\yvec)h^*(\xvec-\yvec)\mu(d\yvec)}, \\
        & \leq \int_{\bbR^n} \abs*{\varphi_{a/4}(\yvec)\left(h^*(\xvec)-h^*(\xvec-\yvec)\right)}\mu(d\yvec),\\
        & \leq \int_{\norm{\yvec}\leq a/4} \varphi_{a/4}(\yvec)\norm{\yvec}\mu(d\yvec), \\
        & \leq \frac{a}{4}\int_{\norm{\yvec}\leq a/4} \varphi_{a/4}(\yvec)\mu(d\yvec), \\
        & = \frac{a}{4},
    \end{align}
    where we have used the fact that $\varphi_{a/4}(\yvec) = 0$ for all $\norm{\yvec} > \frac{a}{4}$, and the fact that $h^*\in\Lip(1)$. We have that
    \begin{align}\label{eq:Lip_mollified_gap}
        & \abs*{\bbE_{\Xvec\sim \nu}[h^*(\Xvec)]\hspace{-1pt} -\hspace{-1pt} \bbE_{\Yvec\sim\rho}[h^*(\Yvec)]} \hspace{-1pt}-\hspace{-1pt} \abs*{\bbE_{\Xvec\sim \nu}[h^*_{a/4}(\Xvec)] \hspace{-1pt}- \hspace{-1pt}\bbE_{\Yvec\sim\rho}[h^*_{a/4}(\Yvec)]} \\
        \begin{split}
        & = \left|\bbE_{\Xvec\sim \nu}[h^*(\Xvec)]\hspace{-1pt} - \hspace{-1pt}\bbE_{\Yvec\sim\rho}[h^*(\Yvec)]\hspace{-1pt} -\hspace{-1pt} \bbE_{\Xvec\sim \nu}[h^*_{a/4}(\Xvec)]\hspace{-1pt} + \hspace{-1pt}\bbE_{\Yvec\sim\rho}[h^*_{a/4}(\Yvec)]\hspace{-1pt} + \hspace{-1pt}\bbE_{\Xvec\sim \nu}[h^*_{a/4}(\Xvec)]\hspace{-1pt}\hspace{-1pt}\right. \\ &\quad\quad\quad\left. -\bbE_{\Yvec\sim\rho}[h^*_{a/4}(\Yvec)]\right| - \abs*{\bbE_{\Xvec\sim \nu}[h^*_{a/4}(\Xvec)] - \bbE_{\Yvec\sim\rho}[h^*_{a/4}(\Yvec)]},
        \end{split}\\
        & \leq \abs*{\bbE_{\Xvec\sim \nu}[h^*(\Xvec) - h^*_{a/4}(\Xvec)]} + \abs*{\bbE_{\Yvec\sim \rho}[h^*(\Yvec) - h^*_{a/4}(\Yvec)]},\\
        & \leq \bbE_{\Xvec\sim \nu}[\abs*{h^*(\Xvec) - h^*_{a/4}(\Xvec)}] + \bbE_{\Yvec\sim \rho}[\abs*{h^*(\Yvec) - h^*_{a/4}(\Yvec)}] \leq \frac{a}{2},
    \end{align}
    where the second to last inequality follows from Jensen's Inequality and the last follows from the fact that $\norm{h^*-h^*_{a/4}}_\infty \leq \frac{a}{4}$. Now, we have from \eqref{eq:Lip_C3_gap} and \eqref{eq:Lip_mollified_gap} that
    \begin{multline}
        \sup_{h \in \calH}\abs*{\bbE_{\Xvec\sim \nu}[h(\Xvec)] - \bbE_{\Yvec\sim\rho}[h(\Yvec)]} + a = \abs*{\bbE_{\Xvec\sim \nu}[h^*(\Xvec)] - \bbE_{\Yvec\sim\rho}[h^*(\Yvec)]} \\ \leq \abs*{\bbE_{\Xvec\sim \nu}[h^*_{a/4}(\Xvec)] - \bbE_{\Yvec\sim\rho}[h^*_{a/4}(\Yvec)]} + a/2,
    \end{multline}
    and thus
    \begin{equation}
        \sup_{h \in \calH}\abs*{\bbE_{\Xvec\sim \nu}[h(\Xvec)] - \bbE_{\Yvec\sim\rho}[h(\Yvec)]} < \abs*{\bbE_{\Xvec\sim \nu}[h^*_{a/4}(\Xvec)] - \bbE_{\Yvec\sim\rho}[h^*_{a/4}(\Yvec)]},
    \end{equation}
    which cannot be true since $h^*_{a/4}\in \calH$. Thus, by contradiction,
    \begin{equation}
        \sup_{h \in \Lip(1)} \abs*{\bbE_{\Xvec\sim \nu}[h(\Xvec)] - \bbE_{\Yvec\sim\rho}[h(\Yvec)]} \leq \sup_{h \in \calH}\abs*{\bbE_{\Xvec\sim \nu}[h(\Xvec)] - \bbE_{\Yvec\sim\rho}[h(\Yvec)]},
    \end{equation}
    which completes the proof.
\end{proof}

\subsection{Proof of Theorem \ref{thm:RRE_GES}}\label{sec:RRE_GES_proof}
We first give a lemma that uses a Foster-Lyapunov criteria \cite{meyn1993stability} to conclude that Condition \ref{cond:moment} holds.
\begin{lemma}\label{lem:FL_condition}
    Consider an SCRN $\Xvec_\Om(t)$ satisfying Assumption \ref{assum:propensities} and \ref{assum:stoich_full_dim}, and suppose there exists $V:\{\xvec\in\bbR^n|\xvec + \vvec^* \in \calX\}\rightarrow\bbR$, $a,b,K >0$, and $\Om_0$ such that $a\norm*{\xvec}^2 \leq V(\xvec) \leq b\norm*{\xvec}^2$ and for all $\Om \in \left\{\Om\in \calO\middle|\Om \geq \Om_0\right\}$ and all $\xvec\in\{\xvec\in\bbR^n|\Om(\xvec + \vvec^*) \in \calX_\Om\}$,
    \begin{equation}
        G_{\tilde\Xvectil_\Om}V(\xvec) \leq - \ga \norm{\xvec}^2 + \frac{K}{\Om}.
    \end{equation}
    Then, Condition \ref{cond:moment} is satisfied.
\end{lemma}
\begin{proof}
    We begin by showing that for all $\Om\in \left\{\Om\in\calO\middle|\Om\leq \Om_0\right\}$, $\Xvec_\Om(t)$ has a unique stationary distribution and is nonexplosive. If $|\calX_\Om|$ is finite, then since $\Xvec_\Om(t)$ is irreducible by Assumption \ref{assum:stoich_full_dim}, $\Xvec_\Om(t)$ has a unique stationary distribution and is nonexplosive. Otherwise, observe that $V(\xvec)$ is radially unbounded due to the assumption that $a\norm*{\xvec}^2 \leq V(\xvec)$. Additionally, by using the bound $V(\xvec) \leq b \norm{\xvec}^2$, we have that
    \begin{equation}
        G_{\tilde\Xvectil_\Om}V(\xvec) \leq - \frac{\ga}{b} V(\xvec) + \frac{K}{\Om}.
    \end{equation}
    Thus, by Theorem 7.1 in \cite{meyn1993stability}, $\tilde\Xvectil_\Om(t)$ is exponentially ergodic, and thus has a unique stationary distribution and is nonexplosive. Now we show the desired moment bounds. We have that for all $\Om \in \left\{\Om \in \calO\middle|\Om\geq \Om_0\right\}$,
    \begin{equation}
        G_{\tilde\Xvectil_\Om}V(\xvec) \leq - \ga( \norm{\xvec}^2+1) +\ga+ \frac{K}{\Om}.
    \end{equation}
    By Theorem 4.3 in \cite{meyn1993stability}, we therefore have that $\bbE\left[\norm*{\tilde\Xvectil_\Om^\infty}^2\right] + 1 \leq \frac{\ga + \frac{K}{\Om}}{\ga}$, which implies that
    \begin{equation}
        \bbE\left[\norm*{\tilde\Xvectil_\Om^\infty}^2\right] \leq \frac{K}{\ga \Om}.
    \end{equation}
    By Jensen's inequality, $\bbE\left[\norm*{\tilde\Xvectil_\Om^\infty}\right] \leq \sqrt{\frac{K}{\ga \Om}}$. Thus, the SCRN satisfies Condition \ref{cond:moment}.
\end{proof}
Building on Lemma \ref{lem:FL_condition}, we now prove lemma that relates the existence of a quadratic type Lyapunov function that for the RRE to Conditions \ref{cond:k_1_moment_finite} and \ref{cond:moment}, by using the Lyapunov function as a Foster-Lyapunov function for $\tilde{\Xvectil}_\Om(t)$.
\begin{lemma}\label{lem:bounded_hessian_lyap}
    Consider an SCRN $\Xvec_\Om(t)$ satisfying Assumptions \ref{assum:propensities}, \ref{assum:C2b}, and \ref{assum:stoich_full_dim}, and let $\vvec^* \in \mathrm{int}(\calX)$ be an equilibrium point of \eqref{eq:LNA_RRE}. Suppose that there exists an open set $\calD \supseteq\left\{\xvec\in \bbR^n \middle| \xvec+\vvec^*\in\calX\right\}$ and a Lyapunov function for the RRE $V:\calD\rightarrow\bbR_{\geq 0}$ twice continuously differentiable and $a,b,\ga >0$ such that for all $\xvec \in \left\{\xvec\in \bbR^n \middle| \xvec+\vvec^*\in\calX\right\}$, $a\norm{\xvec}^2 \leq V(\xvec) \leq b\norm{\xvec}^2$ and
    \begin{equation}
        \frac{\partial V}{\partial \zvec}F(\xvec + \vvec^*) \leq - \gamma \norm{\xvec}^2,
    \end{equation}
    and additionally, there exists $B$ such that $\sup_{\xvec\in\left\{\xvec\in \bbR^n \middle| \xvec+\vvec^*\in\calX\right\}}\norm{\grad^2V(\xvec)} \leq B$. Then, Condition \ref{cond:moment} is satisfied.
\end{lemma}
\begin{proof}
    Let $\eta = \frac{1}{\Om}$, and for $\zvec \in \calX_\Om$ let $\xvec = \eta\zvec - \vvec^*$. Consider 
    \begin{equation}
        G_{\tilde\Xvectil_\Om}V(\xvec) = \sum_{i=1}^r \labar_i(\xvec+\vvec^*) \frac{1}{\eta}\left(V(\xvec+\eta \xivec_i)-V(\xvec)\right).
    \end{equation}
    If we take the Taylor expansion in $\eta$ using the Lagrange form of the remainder we obtain
    \begin{equation}
        G_{\tilde\Xvectil_\Om}V(\xvec) = \sum_{i=1}^r \labar_i(\xvec+\vvec^*) \left( \frac{\partial V}{\partial \xvec}\xivec_i + \frac{1}{2}\eta \xivec_i^T \grad^2{V}(\xvec+ \ep\xivec_i)\xivec_i \right),
    \end{equation}
    where $0 \leq \ep \leq \eta$ depends on $\xvec$. Thus, for sufficiently large $\Om$ we obtain
    \begin{equation}
        G_{\tilde\Xvectil_\Om}V(\xvec) \leq \frac{\partial V(\xvec)}{\partial \zvec} F(\xvec+\vvec^*) + \frac{K}{\Om}(1+\norm{\xvec}^2) \leq - \frac{\ga}{2} \norm{\xvec}^2 + \frac{K}{\Om},
    \end{equation}
    for some $K>0$. Here we have used the fact that by Assumption \ref{assum:C2b}, there exists $K'$ such that $\abs*{\sum_{i=1}^r \labar_i(\xvec+\vvec^*)}\leq K'(1+\norm{\xvec}^2)$. Hence, by Lemma \ref{lem:FL_condition}, Condition \ref{cond:moment} is satisfied.
\end{proof}
We now prove Theorem \ref{thm:RRE_GES}. We note that we only need to verify Condition \ref{cond:k_1_moment_finite} for $\kappa = 2$, since we are assuming that the Hessian of $\labar_i(\Xvec)$ is bounded.
\begin{proof}
We begin by checking that Condition \ref{cond:k_1_moment_finite} holds. Let us consider the scaled Markov chain $\Zvec_\Om(t) = \frac{1}{\Om}\Xvec_\Om(t)$ and the function $V(\vvec) = \left(\cvec^T\vvec\right)^{\kappa + 1}$. We have that
\begin{align}
    G_{\Zvec_\Om}V(\vvec) & = \sum_{i=1}^r \la_i\left(\Om\vvec\right) \left( V(\vvec + \frac{1}{\Om}\xivec_i) - V(\vvec) \right)\\
    & = \sum_{i=1}^r \Om \labar_i\left(\vvec\right) \left( V(\vvec + \frac{1}{\Om}\xivec_i) - V(\vvec) \right)\\
    & = \sum_{i=1}^r \Om \labar_i\left(\vvec\right) \left( \left(\cvec^T(\vvec + \frac{1}{\Om}\xivec_i)\right)^{\kappa+1} - \left(\cvec^T\vvec\right)^{\kappa+1} \right)
\end{align}
For $\kappa = 2$, we have 
\begin{align}
    \begin{split}
    \Om \left( \left(\cvec^T(\vvec + \frac{1}{\Om}\xivec_i)\right)^{\kappa+1} - \left(\cvec^T\vvec\right)^{\kappa+1} \right) & = \Om \left( (\cvec^T\vvec)^3 + \frac{3}{\Om}(\cvec^T\vvec)^2\cvec^T\xivec_i \right. \\ &\left. + \frac{3}{\Om^2} \cvec^T\vvec (\cvec^T\xivec_i)^2 + \frac{1}{\Om^3}(\cvec^T\xivec_i)^3 -(\cvec^T\vvec)^3\right)
    \end{split}\\
    & = 3(\cvec^T\vvec)^2\cvec^T\xivec_i + \frac{3}{\Om} \cvec^T\vvec (\cvec^T\xivec_i)^2 + \frac{1}{\Om^2}(\cvec^T\xivec_i)^3
\end{align}
Noting that $\frac{\partial}{\partial \vvec}V(\vvec) = (\kappa+1)(\cvec^T\vvec)^\kappa \cvec^T$, we have
\begin{align}
    G_{\Zvec_\Om}V(\vvec) & = \sum_{i=1}^r \labar_i\left(\vvec\right) \left(3(\cvec^T\vvec)^2\cvec^T\xivec_i + \frac{3}{\Om} \cvec^T\vvec (\cvec^T\xivec_i)^2 + \frac{1}{\Om^2}(\cvec^T\xivec_i)^3\right)\\
    & = \frac{\partial V(\vvec)}{\partial \vvec} F(\vvec) + \frac{3}{\Om}\sum_{i=1}^r \labar_i\left(\vvec\right)\cvec^T\vvec (\cvec^T\xivec_i)^2 + \frac{1}{\Om^2}\sum_{i=1}^r\labar_i\left(\vvec\right)(\cvec^T\xivec_i)^3
\end{align}
Thus, for sufficiently large $\Om$, and $\vvec \geq 0$ such that $\cvec^T\vvec \geq d$, we have that $G_{\Zvec_\Om}V(\vvec) \leq - \frac{1}{2}\ga_2 (\cvec^T\vvec)^3$. Therefore, 
\begin{equation}
    G_{\Zvec_\Om}V(\vvec) \leq -\left( \frac{1}{2}\ga_2 (\cvec^T\vvec)^3 + 1\right) + \max\limits_{\vvec\in\frac{1}{\Om}\calX_\Om: \cvec^T\vvec \leq d} \left(G_{\Zvec_\Om}V(\vvec) + \frac{1}{2}\ga_2 (\cvec^T\vvec)^3 + 1\right),
\end{equation}
which shows by Theorem 7.1 in \cite{meyn1993stability} that $\Zvec_\Om(t)$ has a unique stationary distribution and that by Theorem 4.3 in \cite{meyn1993stability} $\bbE[(\cvec^T\Zvec_\Om^\infty)^3] < \infty$, and thus $\tilde{\Xvectil}_\Om(t)$ has a unique stationary distribution and $\bbE[(\cvec^T\tilde{\Xvectil}_\Om^\infty)^3] < \infty$. Therefore, Condition \ref{cond:k_1_moment_finite} is satisfied. We now verify that Condition \ref{cond:moment} holds. We will need the following Lemma, which is a slight variation on the standard converse Lyapunov theorem, see e.g. \cite{khalil2002nonlinear}.
\begin{lemma}\label{lem:converse_lyap}
    Let $\calX \subseteq \bbR^n_{\geq 0}$ be closed and let $\calD \subset \bbR^n$ be an open set such that $\calX \subset \calD$. Let $\bar{F}:\calD\rightarrow\bbR^n$ be a $C^2$ function defining a vector field for which $\calX$ is positively invariant. Suppose that for all $\xvec \in \calX$, the solution to
    \begin{equation}
        \dot{\xvec}' = \bar{F}(\xvec'),\quad \xvec'(0)=\xvec,
    \end{equation}
    $\phi(t;\xvec)$, exists uniquely on $[0,\infty)$ and that there exists $\bar{K},\gabar_1 >0$ and $\xvec^* \in \calX$ such that for all $\xvec \in \calX$ and all $t\geq 0$,
    \begin{equation}
        \norm{\phi(t;\xvec)-\xvec^*} \leq \bar{K}e^{-\gabar_1 t}\norm{\xvec-\xvec^*},
    \end{equation}
    and furthermore, there exist $\cvec \in \bbR^n_{>0}$, $\bar{d}>0,\gabar_2>0$ such that for all $\xvec\in\left\{\xvec \in \calX \middle| \bar{\cvec}^T\xvec \geq \bar{d} \right\}$,
    \begin{equation}
        \bar{\cvec}^T\bar{F}(\xvec) \leq -\ga_2 \bar{\cvec}^T \xvec.
    \end{equation}
    Then, for each $d'\geq \bar{d}$, there exists $\overline{W}_1:\calD\rightarrow\bbR$ twice continuously differentiable and $\bar{a}_1, \bar{b}_1,\gabar >0$ such that for all $\xvec \in \left\{\xvec \in \calX \middle| \bar{\cvec}^T\xvec \leq d' \right\}$,
    \begin{align}
        \bar{a}_1 \norm{\xvec-\xvec^*}^2 \leq \overline{W}_1(\xvec) & \leq \bar{b}_1 \norm{\xvec-\xvec^*}^2,\\
    \frac{\partial \overline{W}_1 (\xvec)}{\partial \vvec}\bar{F}(\xvec) & \leq - \gabar \norm*{\xvec-\xvec^*}^2.
    \end{align}
\end{lemma}

\begin{proof}
    The proof that there exists $\overline{W}_1$ satisfying the quadratic bounds proceeds analogously to the proof of Theorem 4.14 in \cite{khalil2002nonlinear}, with the only adjustment that we use $\left\{\xvec\in \calX\middle|\cvec^T\xvec \leq d'\right\}$ as the domain instead of a ball around the origin. That $\overline{W}_1 \in C^2$ is assured by the form
    \begin{equation}
        \overline{W}_1(\xvec) = \int_0^s \phi(t;\xvec)^T \phi(t;\xvec) dt,
    \end{equation}
    where $\phi(t;\xvec)$ is the solution to $\dot{\xvec}' = \bar{F}(\xvec')$ with $\xvec'(0) = \xvec$ and $\phi(t;\xvec)\in C^2$ by $\bar{F}\in C^2$ and e.g. Theorem 8.43 in \cite{walter2010theory}.
\end{proof}

We now use Lemma \ref{lem:converse_lyap} to construct a Lyapunov Function $W_1$. Recall that $\gamma_1$, $K$, $\cvec$, and $d$ are assumed to exist by the assumptions of Theorem \ref{thm:RRE_GES}. We apply Lemma \ref{lem:converse_lyap} to $\bar{F}(\xvec) = F(\xvec)$ with $\gabar_1 = \gabar$, $\bar{K}=K$, $\bar{\cvec}=\cvec$, $\bar{d} = d$, $\xvec^*=\vvec^*$, and $\gabar_2 = \ga_2$, picking $d'>d$ to obtain a function $\overline{W}_1$ with associated values $\bar{a}_1$, $\bar{b}_1$, and $\gabar$. Let $W_1(\vvec) = \overline{W}_1(\vvec)$. In order to create a Lyapunov function on all of $\calX$, we define $W_2:\calD\rightarrow\bbR$ as
\begin{equation}
    W_2(\vvec) = 2\bar{b}_1\left(\left(\cvec^T\vvec\right)^2 - (d')^2\right).
\end{equation}
We construct a global Lyapunov function by merging $W_1$ near the origin with $W_2$ far from the origin. We must do this in a way that creates a $C^2$ function with bounded second derivative so that we can apply Lemma \ref{lem:bounded_hessian_lyap}. Our approach is similar to a technique used to prove Theorem B.1 in \cite{lin1996smooth}. To this end, let $W:\calD\rightarrow\bbR$ be defined by
\begin{equation}
    W(\vvec) = \max\left\{W_1(\vvec), W_2(\vvec)\right\}.
\end{equation}
We have that for all $\vvec\in\left\{\vvec\in\calX \middle| \cvec^T\vvec < d'\right\}$, $W(\vvec) = W_1(\vvec)$, and hence in this region,
\begin{align}
    \frac{\partial W (\vvec)}{\partial \vvec}F(\vvec) & \leq - \gabar \norm*{\vvec-\vvec^*}^2.
\end{align}
Observe that due to the quadratic upper bound on $W_1(\vvec)$, there exists $d'' > d'$ such that for all $\vvec\in \left\{\vvec\in\calX \middle| \cvec^T\vvec > d''\right\}$, $W_2(\vvec)> W_1(\vvec)$, and thus in this region
\begin{align}
    \frac{\partial W (\vvec)}{\partial \vvec}F(\vvec) & \leq - 4\bar{b}_1\ga_2 (\cvec^T\vvec)^2,\\
    & \leq - 4\bar{b}_1\ga_2 \sum_i c_i^2 v_i^2,\\
    & \leq - 4\bar{b}_1\ga_2 \left(\min_i c_i^2 \right) \sum_i v_i^2,\\
    & \leq - 4\bar{b}_1\ga_2 \left(\min_i c_i^2 \right) \sum_i (v_i^2 - 2v_iv_i^*),
\end{align}
where we have used that $\calX \subseteq \bbR_{\geq 0}^n$ in the last line. hence, there exists $\tilde{\ga}_2 >0$ such that for all $\vvec \in \left\{\vvec\in\calX \middle| \cvec^T\vvec > d''\right\}$,
\begin{align}
    \frac{\partial W (\vvec)}{\partial \vvec}F(\vvec) & \leq - \tilde{\ga}_2 \sum_i (v_i^2 - 2v_iv_i^* + (v^*_i)^2),\\
    & = - \tilde{\ga}_2 \norm*{\vvec-\vvec^*}^2.
\end{align}

Let $d_1,d_2,d_3,d_4$ be such that $\cvec^T\vvec^* < d_1 < d_2 < d'$ and $d'' < d_3 < d_4$. Let $S_1 =  \left\{\vvec \in \calX\middle| \cvec^T\vvec < d_2\right\}$, $S_2 = \left\{\vvec \in \calX\middle| d_1 < \cvec^T\vvec < d_4\right\}$, and $S_3 = \left\{\vvec \in \calX\middle| d_3 < \cvec^T\vvec \right\}$. Now consider $\vvec\in S_2$. Let $\varphi_\ep(\xvec)$ be the bump function defined by
    \begin{equation}
        \varphi_\ep (\xvec) = \left\{\begin{array}{cc}
             \frac{1}{\ep^n}Ce^{\frac{-1}{1 - \norm{\xvec/\ep}^2}}, & \norm{\xvec/\ep}\leq 1, \\
             0, & \text{else},
        \end{array}\right.
    \end{equation}
    where $C = \left(\int_{\bbR^n} e^{\frac{-1}{1 - \norm{\xvec}^2}}d\xvec\right)^{-1}$, and let $W_\ep (\vvec) = W \star \varphi_\ep$, where we treat $W(\vvec) = 0$ for all $\vvec\notin \calD$. We have that
\begin{align}
    \frac{\partial W_\ep (\vvec)}{\partial \vvec} F(\vvec) & = \frac{\partial}{\partial \vvec} \int_{\bbR^n}W(\vvec-\svec)\varphi_\ep(\svec)\mu(d\svec) \cdot F(\vvec),
\end{align}
where $\mu$ is the Lebesgue measure on $\bbR^n$. For $\tauvec \in \bbR^n$, we have that
\begin{equation}
    \frac{\partial}{\partial \vvec} W(\vvec-\svec)\varphi_\ep(\svec) \cdot \tauvec = \lim_{\eta \rightarrow 0} \frac{W(\vvec-\svec + \eta \tauvec) - W(\vvec-\svec)}{\eta}\varphi_\ep(\svec),
\end{equation}
and for $\eta >0$ we have
\begin{equation}
    \norm*{\frac{W(\vvec-\svec + \eta \tauvec) - W(\vvec-\svec)}{\eta}\varphi_\ep(\svec)} \leq L\norm{\tauvec} \cdot \max_{\svec\in \bbR^n}\norm{\varphi_\ep(\svec)},
\end{equation}
where $L$ is the Lipschitz constant of $W$ on $S_2$, and $\max_{\svec\in \bbR^n}\norm{\varphi_\ep(\svec)} = C/(e\ep^n)$. By the Dominated Convergence Theorem and the fact that $\lim_{\eta\rightarrow 0} \frac{W(\vvec-\svec + \eta \tauvec) - W(\vvec-\svec)}{\eta}\varphi_\ep(\svec) = \frac{\partial W(\vvec)}{\partial \vvec}\varphi_\ep(\svec)\cdot \tauvec$ almost everywhere,
\begin{equation}
     \frac{\partial W_\ep (\vvec)}{\partial \vvec} F(\vvec) =  \int_{\bbR^n} \frac{\partial W(\vvec-\svec)}{\partial \vvec} \cdot F(\vvec) \varphi_\ep(\svec)\mu(d\svec).
\end{equation}
For sufficiently small $\ep$, there exists $\ga > 0$ such that for all $\norm{\svec} \leq \ep$ and $\vvec\in S_2$ such that $W$ is differentiable at $\vvec-\svec$,
\begin{equation}
    \frac{\partial W(\vvec-\svec)}{\partial \vvec}F(\vvec) \leq -\ga \norm{\vvec-\vvec^*}^2.
\end{equation}
Thus, for $\vvec \in S_2$, we have
\begin{equation}
    \frac{\partial W_\ep (\vvec)}{\partial \vvec} F(\vvec) \leq -\ga \norm*{\vvec-\vvec^*}^2.
\end{equation}
Now, let $\beta_i:\calD \rightarrow [0,1]$ for $i=\{1,2,3\}$ be a smooth partition of unity subordinate to $\{S_1,S_2,S_3\}$, and let
\begin{equation}
    W_s(\vvec) = \beta_1(\vvec)W_1(\vvec) + \beta_2(\vvec)W_\ep(\vvec) + \beta_3(\vvec)W_2(\vvec).
\end{equation}
We have that
\begin{align}
    \begin{split}
    \frac{\partial W_s(\vvec)}{\partial \vvec}F(\vvec) & = \left(\beta_1(\vvec)\frac{\partial W_1(\vvec)}{\partial \vvec} + \beta_2(\vvec)\frac{\partial W_\ep(\vvec)}{\partial \vvec} + \beta_3(\vvec)\frac{\partial W_3(\vvec)}{\partial \vvec}\right) F(\vvec)
    \\&\quad + \left( \frac{\partial \beta_1(\vvec)}{\partial \vvec} W_1(\vvec) + \frac{\partial \beta_2(\vvec)}{\partial \vvec} W_\ep(\vvec) + \frac{\partial \beta_3(\vvec)}{\partial \vvec} W_2(\vvec) \right)F(\vvec),
    \end{split}\\
    \begin{split}
    \frac{\partial W_s(\vvec)}{\partial \vvec}F(\vvec) & = \left(\beta_1(\vvec)\frac{\partial W_1(\vvec)}{\partial \vvec} + \beta_2(\vvec)\frac{\partial W_\ep(\vvec)}{\partial \vvec} + \beta_3(\vvec)\frac{\partial W_3(\vvec)}{\partial \vvec}\right) F(\vvec)
    \\&\quad + \left( \frac{\partial \beta_1(\vvec)}{\partial \vvec} (W_1(\vvec)-W(\vvec)) + \frac{\partial \beta_2(\vvec)}{\partial \vvec} (W_\ep(\vvec)-W(\vvec)) \right.\\& \left. + \frac{\partial \beta_3(\vvec)}{\partial \vvec} (W_2(\vvec)-W(\vvec)) \right)F(\vvec),
    \end{split}
\end{align}
where we have used the fact that $\sum_{j=1}^3 \frac{\partial \beta_j(\vvec)}{\partial \vvec} = 0$.

Let us consider $\vvec \in \calX \cap (S_1 \setminus S_2)$. In this region, $W(\vvec) = W_1(\vvec)$, $\beta_2(\vvec)=\beta_3(\vvec) = 0$, and $\frac{\partial \beta_2(\vvec)}{\partial \vvec}=\frac{\partial \beta_3(\vvec)}{\partial \vvec} = 0$. Thus,
\begin{equation}
    \frac{\partial W_s(\vvec)}{\partial \vvec}F(\vvec) \leq -\ga_1 \norm*{\vvec-\vvec^*}^2.
\end{equation}

Now, let us consider $\vvec \in \overline{\calX \cap S_1 \cap S_2}$. In this region, $W(\vvec) = W_1(\vvec)$, $\beta_3(\vvec)=0$, and $\frac{\partial \beta_3(\vvec)}{\partial \vvec}= 0$, implying that
\begin{equation}
    \frac{\partial W_s(\vvec)}{\partial \vvec}F(\vvec)  \leq -\min\{\ga,\ga_1\} \norm*{\vvec-\vvec^*}^2  + \norm*{\frac{\partial \beta_2(\vvec)}{\partial \vvec} (W_\ep(\vvec)-W(\vvec))}.
\end{equation}
By e.g. Lemma B.2 in \cite{lin1996smooth}, $W_\ep(\vvec)\rightarrow W(\vvec)$ uniformly on the compact set $\overline{\calX \cap S_1 \cap S_2}$. Additionally, observe that $\min\limits_{\vvec \in\overline{\calX \cap S_1 \cap S_2}} \norm*{\vvec-\vvec^*}^2 > 0$, and thus for sufficiently small $\ep$, for all $\vvec \in \overline{\calX \cap S_1 \cap S_2}$,
\begin{equation}
    \frac{\partial W_s(\vvec)}{\partial \vvec}F(\vvec)  \leq - \frac{1}{2} \min\{\ga,\ga_1\} \norm*{\vvec-\vvec^*}^2.
\end{equation}

Now, let us consider $\vvec \in \overline{\calX \cap (S_2\setminus S_1 \setminus S_3)}$. In this region, $\beta_1(\vvec) = \beta_3(\vvec) = 0$ and $\frac{\partial \beta_1(\vvec)}{\partial \vvec} = \frac{\partial \beta_3(\vvec)}{\partial \vvec} = 0$, and thus
\begin{equation}
    \frac{\partial W_s(\vvec)}{\partial \vvec}F(\vvec)  \leq -\ga \norm*{\vvec-\vvec^*}^2  + \norm*{\frac{\partial \beta_2(\vvec)}{\partial \vvec} (W_\ep(\vvec)-W(\vvec))}.
\end{equation}
By noting that $\overline{\calX \cap (S_2\setminus S_1 \setminus S_3)}$ is compact, and again using Lemma B.2 from \cite{lin1996smooth} and the fact that\\ $\min\limits_{\vvec\in \overline{\calX \cap (S_2\setminus S_1 \setminus S_3)}} \norm*{\vvec-\vvec^*}^2 > 0$, we have that for sufficiently small $\ep$,
\begin{equation}
    \frac{\partial W_s(\vvec)}{\partial \vvec}F(\vvec)  \leq -\frac{1}{2} \ga \norm*{\vvec-\vvec^*}^2.
\end{equation}

Let us consider $\vvec \in \overline{\calX \cap (S_3 \cap S_2)}$. In this region, $\bt_1(\vvec)=0$ and $W(\vvec) = W_2(\vvec)$. Thus,
\begin{equation}
    \frac{\partial W_s(\vvec)}{\partial \vvec}F(\vvec)  \leq -\min\{\ga,\tilde{\ga}_2\} \norm*{\vvec-\vvec^*}^2 + \norm*{\frac{\partial \beta_2(\vvec)}{\partial \vvec} (W_\ep(\vvec)-W(\vvec))}.
\end{equation}
Again using Lemma B.2 from \cite{lin1996smooth}, and the fact that $\min\limits_{\vvec \in\overline{\calX \cap (S_3 \cap S_2)}} \norm*{\vvec-\vvec^*}^2 > 0$, we have that for sufficiently small epsilon,
\begin{equation}
    \frac{\partial W_s(\vvec)}{\partial \vvec}F(\vvec)  \leq -\frac{1}{2}\min\{\ga,\tilde{\ga}_2\} \norm*{\vvec-\vvec^*}^2.
\end{equation}

Finally, consider $\vvec \in \overline{\calX \cap (S_3 \setminus S_2)}$. In this region, $W(\vvec) = W_2(\vvec)$, $\bt_1(\vvec)=\bt_2(\vvec)=0$, and $\frac{\partial \beta_1(\vvec)}{\partial \vvec} = \frac{\partial \beta_2(\vvec)}{\partial \vvec} = 0$. Thus,
\begin{equation}
    \frac{\partial W_s(\vvec)}{\partial \vvec}F(\vvec)  \leq - \tilde{\ga}_2 \norm*{\vvec-\vvec^*}^2.
\end{equation}
Combining the preceding results for each subset of $\calX$, we have that for sufficiently small $\ep$, for all $\vvec\in \calX$,
\begin{equation}
     \frac{\partial W_s(\vvec)}{\partial \vvec}F(\vvec)  \leq - \frac{1}{2}\min\{\ga,\ga_1,\tilde{\ga}_2\} \norm*{\vvec-\vvec^*}^2.
\end{equation}
Observe that
\begin{equation}
    \sup_{\vvec\in\calX} \norm*{\grad^2{W_s}(\vvec)} = \max \left\{ \max_{\vvec \in \overline{\calX \setminus S_1 \setminus S_2}} \norm*{\grad^2{W_s}(\vvec)}, 4\bar{b}_1 \norm*{\cvec\cvec^T} \right\},
\end{equation}
where $\max_{\vvec \in \overline{\calX \setminus S_1 \setminus S_2}} \norm*{\grad^2{W_s}(\vvec)}$ exists by the smoothness of $W_s$ and the compactness of $\overline{\calX \setminus S_1 \setminus S_2}$. Therefore, for some sufficiently small $\ep$, we can apply Lemma  \ref{lem:bounded_hessian_lyap} with $V(\vvec) = W_s(\vvec)$ to conclude that the SCRN satisfies Condition \ref{cond:moment}.

We have now verified Conditions \ref{cond:k_1_moment_finite} and \ref{cond:moment}. Observe that $0 = F(\vvec^*)$, i.e. $\vvec^*$ is an equilibrium point of \eqref{eq:LNA_RRE}, and by Corollary 4.3 in \cite{khalil2002nonlinear}, the global exponential stability of $\vvec^*$ implies that $\frac{\partial F(\vvec^*)}{\partial \vvec}$ is Hurwitz. Thus, we can apply Theorem \ref{thm:moments_imply_convergence} to obtain the desired result.
\end{proof}

\section{Examples}\label{sec:examples}
\subsection{Antithetic motif}
The antithetic feedback motif is a common SCRN encountered in synthetic biology because it can approximately implement an integrator \cite{qian2018realizing,briat2016antithetic}. The antithetic motif is given by the following set of chemical reactions:
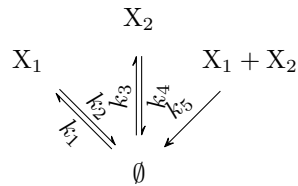
\begin{equation}\label{eq:ex1_reactions}
	\vcenter{\hbox{\schemestart
		\subscheme{$\emptyset$}
    		\arrow(z--x1){<=>[$k_2$][$k_1$]}[135]
		\subscheme{$\mathrm{X}_1$}
		\arrow(@z--x2){<=>[$k_3$][$k_4$]}[90]
		\subscheme{$\mathrm{X}_2$}
		\arrow(@z--x12){<-[$k_5$]}[45]
		\subscheme{$\mathrm{X}_1 + \mathrm{X}_2$}
	\schemestop}}
\end{equation}
The reaction vectors and propensities are:
\begin{align}
    \la_1(\xvec) & = \Om k_1, & \xivec_1 = \begin{bmatrix}+1 & 0 \end{bmatrix}^T,\\
    \la_2(\xvec) & = k_2x_1, & \xivec_2 = \begin{bmatrix}-1 & 0 \end{bmatrix}^T,\\
    \la_3(\xvec) & = \Om k_3, & \xivec_3 = \begin{bmatrix} 0 & +1 \end{bmatrix}^T,\\
    \la_4(\xvec) & = k_4x_2, & \xivec_4 = \begin{bmatrix} 0 & -1 \end{bmatrix}^T,\\
    \la_5(\xvec) & = \frac{k_5}{\Om}x_1x_2, & \xivec_5 = \begin{bmatrix} -1 & -1 \end{bmatrix}^T.\\
\end{align}
This SCRN satisfies Assumptions \ref{assum:propensities}, \ref{assum:C2b}, and \ref{assum:stoich_full_dim}, with $\calX_\Om = \bbZ_{\geq 0}^2$. We show that we can apply Theorem \ref{thm:RRE_GES}. Observe that as shown in \cite{blanchini2011structurally}, the RREs
\begin{align}
    \dot{v}_1 & = k_1 - k_2v_1 - k_5v_1v_2,\\
    \dot{v}_2 & = k_3 - k_4v_2 - k_5v_1v_2,
\end{align}
have a unique, globally exponentially stable equilibrium point $\vvec^*$ in $\bbR^n$, with $\vvec^* > 0$. Letting $\cvec = \begin{bmatrix} 1 & 1 \end{bmatrix}^T$ we have that $\cvec^TF(\vvec) = k_1 + k_2 - k_2v_1 - k_4v_2 - 2k_5v_1v_2 \leq k_1+k_2 - \min\{k_2,k_4\}(v_1+v_2)$, and thus there exists $d,\ga>0$ such for all $\vvec \in \left\{\vvec \in \bbR_{\geq 0}^2 \middle| \cvec^T\vvec \geq d \right\}$, $\cvec^TF(\vvec) \leq -\ga \cvec^T\vvec$. Therefore, by Theorem \ref{thm:RRE_GES}, there exists $C,\Om_0'>0$ such that for all $\Omega \geq \Omega_0'$,
    \begin{equation}
        \Wass\left(\tilde{\Xvec}_\Om^\infty, \Yvec^\infty\right)
        \leq C \frac{\ln\Omega}{\sqrt{\Omega}}.
    \end{equation}

\subsection{Transcriptional activation}
We consider a simple model of a protein $P$ that is transcriptionally regulated by $R$ binding to the promoter $T$ to form complex $C$ \cite{BFS}:
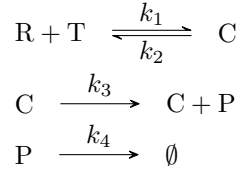
\begin{align}\label{eq:ex2_reactions}
	&\vcenter{\hbox{\schemestart
		\subscheme{$\mathrm{R}+\mathrm{T}$}
    		\arrow(z--x1){<=>[$k_1$][$k_2$]}[0]
		\subscheme{$\mathrm{C}$}
	\schemestop}}\\
    &\vcenter{\hbox{\schemestart
		\subscheme{$\mathrm{C}$}
    		\arrow(z--x1){->[$k_3$]}[0]
		\subscheme{$\mathrm{C}+\mathrm{P}$}
	\schemestop}}\\
    &\vcenter{\hbox{\schemestart
		\subscheme{$\mathrm{P}$}
    		\arrow(z--x1){->[$k_4$]}[0]
		\subscheme{$\emptyset$}
	\schemestop}}
\end{align}
This SCRN has four species, however, we can use the two conservation laws $R+C = R_{tot}$ and $T+C = T_{tot}$ to obtain the following reduced model:
\begin{align}
    X_1 & = C\\
    X_2 & = P
\end{align}
with reactions
\begin{align}
    \la_1(\xvec) & = \frac{k_1}{\Om}(R_{tot}-x_1)(T_{tot}-x_1), & \xivec_1 = \begin{bmatrix}+1 & 0 \end{bmatrix}^T,\\
    \la_2(\xvec) & = k_2x_1, & \xivec_2 = \begin{bmatrix}-1 & 0 \end{bmatrix}^T,\\
    \la_3(\xvec) & = k_3x_1, & \xivec_3 = \begin{bmatrix} 0 & +1 \end{bmatrix}^T,\\
    \la_4(\xvec) & = k_4x_2, & \xivec_4 = \begin{bmatrix} 0 & -1 \end{bmatrix}^T,\\
\end{align}
resulting in the RREs
\begin{align}
    \dot{v}_1 & = k_1(\bar{R}_{tot}-v_1)(\bar{T}_{tot}-v_1) - k_2v_1,\\
    \dot{v}_2 & = k_3v_1 - k_4v_2.
\end{align}
If we set $R_{tot} = \Om \bar{R}_{tot}$ and $T_{tot} = \Om\bar{T}_{tot}$, with $\bar{R}_{tot},\bar{T}_{tot} \in \bbQ$, the SCRN satisfies Assumptions \ref{assum:propensities} and \ref{assum:C2b} with $\calX = [0, \min\{\bar{R}_{tot},\bar{T}_{tot}\}]\times [0, \infty)$. Additionally, the SCRN satisfies Assumption \ref{assum:stoich_full_dim}. The RREs have a unique, globally exponentially stable equilibrium point in $\calX$, and letting $\cvec = \begin{bmatrix} 0 & 1 \end{bmatrix}^T$, we have $\cvec^TF(\vvec) = k_3v_1 - k_4v_2$, and so there exists $\ga,d>0$ such that for all $\vvec \in \left\{\vvec\in\calX\middle|\cvec^T\vvec \geq d\right\}$, $\cvec^TF(\vvec) \leq - \ga \cvec^T\vvec$. Therefore, by Theorem \ref{thm:RRE_GES}, there exists $C,\Om_0'>0$ such that for all $\Om\in\left\{\Om\in\calO\middle|\Omega \geq \Omega_0'\right\}$,
    \begin{equation}
        \Wass\left(\tilde{\Xvec}_\Om^\infty, \Yvec^\infty\right)
        \leq C \frac{\ln\Omega}{\sqrt{\Omega}}.
    \end{equation}

\section{Conclusion}\label{sec:conclusion}
In this work we studied the relationship between the stationary distributions of appropriately scaled Markov chains corresponding to SCRNs, and the steady state behavior of the RRE and LNA models which are commonly used as approximations. Using Stein's Method, we obtained a bound on the 1-Wasserstein distance between $\Xvectil_\Om^\infty = \frac{1}{\sqrt{\Om}}\left(\Xvec_\Om^\infty - \Om\vvec^*\right)$ and $\Yvec^\infty$, the fluctuation term of the LNA, that is proportional to $\frac{\ln \Om}{\sqrt{\Om}}$, and a bound on the 1-Wasserstein distance between $\frac{1}{\Om}\Xvec_\Om^\infty$ and $\vvec^*$ that is proportional to $\frac{1}{\sqrt{\Om}}$. Our main result requires one to control the second moment of $\Xvectil_\Om^\infty$, and we presented a Foster-Lyapunov based method to do so by analyzing just the RREs. Our condition requires global exponential stability of $\vvec^*$ in the RRE, along with the existence of a linear function with exponential decay far from the origin. While proving the global exponential stability of an equilibrium point of the RRE is in and of itself a challenging task, a large body of literature on the topic is available, since the RRE have been the subject on intense study over the years. In fact, recent advances in global stability of equilibria of the RRE can potentially be leveraged via our results to analyze the error in the RRE and LNA for large volume \cite{al2015new,al2023graphical}. The authors foresee this work applying to model reduction via timescale separation of SCRNs, and system identification algorithms, where the computational simplicity of the LNA makes it an attractive model. In both cases, the fact that we can obtain non-asymptotic error bounds on the LNA is critical to analyzing the error in the reduced models and the identified system respectively.

\section{Acknowledgments}
This work was supported in part by the U.S. AFOSR under grant FA9550-22-1-0316. The authors thank Dylan Hirsch for his comments on the proofs of Lemma \ref{lem:wasserstein_C3} and Theorem \ref{thm:RRE_GES}, and thank Ruth J Williams for her helpful feedback on a preliminary version of this work.

\bibliographystyle{plain} %
\bibliography{myReferences}       %

\end{document}